\documentclass[journal]{IEEEtran}
\makeatletter
\newcommand*{\rom}[1]{\expandafter\@slowromancap\romannumeral #1@}
\makeatother
\IEEEoverridecommandlockouts
\ifCLASSINFOpdf
\else
\fi

\usepackage{stfloats}
\usepackage{float}
\usepackage{subfloat}
\usepackage{graphics}
\usepackage{multicol}
\usepackage{lipsum}
\usepackage[cmex10]{amsmath}
\usepackage{amsthm}
\usepackage{mathrsfs}
\usepackage{mathtools}
\usepackage{amsbsy}
\usepackage{xcolor}
\usepackage{mathrsfs}
\usepackage{setspace}
\usepackage{graphicx}
\usepackage{lettrine}
\usepackage{newclude}
\usepackage[normalem]{ulem}
\usepackage{latexsym}
\usepackage{algpseudocode}
\usepackage{algorithm,algpseudocode}
\usepackage{algorithmicx}

\usepackage{multirow}
\usepackage{rotating}
\usepackage{booktabs}
\usepackage{wasysym}
\usepackage{mathrsfs}
\usepackage{bbm} 
\usepackage{filecontents}
\usepackage{soul}

\usepackage{optidef}

\newtheorem{theorem}{Theorem}
\newtheorem{lemma}{Lemma}

\newtheorem{definition}{Definition}

\ifCLASSOPTIONcompsoc
\usepackage[caption=false,font=normalsize,labelfon
t=sf,textfont=sf]{subfig}
\else
\usepackage[caption=false,font=footnotesize]{subfig}
\fi
\def\BibTeX{{\rm B\kern-.05em{\sc i\kern-.025em b}\kern-.08em
    T\kern-.1667em\lower.7ex\hbox{E}\kern-.125emX}}
\usepackage{amsmath,epsfig,cite,amsfonts,amssymb,psfrag,color}
\usepackage{epstopdf}

\def\BibTeX{{\rm B\kern-.05em{\sc i\kern-.025em b}\kern-.08em
    T\kern-.1667em\lower.7ex\hbox{E}\kern-.125emX}}

\newcommand{\abssOne}[1]{{\left\lvert{#1}\right\rvert}}

\newcommand{\abssSq}[1]{{\left\lvert{#1}\right\rvert}^2}

\newcommand{\abs}[1]{{\left\lvert{#1}\right\rvert}}

\newcommand{\norm}[1]{\left\lVert#1\right\rVert}

\begin{document}
\bstctlcite{IEEEexample:BSTcontrol}

	\title{Characterization of Beam-Squint and Beam-Split Effects in RIS-assisted Multi-Frequency Networks}
	\author{Mohammad Amin Saeidi, {\em Member IEEE},  Hina Tabassum, {\em Senior Member IEEE} 
    \thanks{
        M. A. Saeidi and H. Tabassum are with the Department of Electrical Engineering and Computer Science, York University, Toronto, Canada (emails:  
        \{{amin96a}, {hinat}\}@yorku.ca). This work was supported  by Natural Sciences and Engineering
Research Council of Canada (NSERC) Discovery grant.
        
	}}
\raggedbottom

\maketitle

\begin{abstract}
Multi-frequency operation in beyond-5G and 6G systems renders the beam directions of large-aperture arrays inherently frequency-dependent, giving rise to beam misalignment effects that are particularly critical for reconfigurable intelligent surfaces (RISs). When an RIS phase profile is configured at one frequency but illuminated at another, the reflected field may exhibit main-lobe deviation and the formation of additional dominant lobes, degrading beamforming performance. In this paper, we develop a unified analytical framework to rigorously characterize these effects for uniform planar array (UPA)-based RISs under both continuous and practical finite-resolution phase control. We formalize beam-squint as the deviation of the \textcolor{black}{maximum-gain} beam from the desired elevation–azimuth angle pair and beam-split as the emergence of additional \textcolor{black}{maximum-gain} beams under frequency mismatch and quantization. For continuous-phase RISs, we derive necessary and sufficient peak conditions, obtain closed-form elevation–azimuth peak families, and establish explicit feasibility conditions as functions of the frequency ratio, incidence and configuration angles, and element spacing. Extending the analysis to practical $b$-bit phase quantization, we express the quantized beampattern as a superposition of harmonic array responses via Fourier expansion of the quantization-error phasor, leading to tractable closed-form per-harmonic peak families, dominance ordering, and feasibility conditions. To capture peak displacement induced by harmonic superposition, we further develop a perturbation-based correction leveraging dominant-harmonic curvature and residual-harmonic gradients, improving peak-location accuracy without exhaustive two-dimensional searches. Numerical simulations validate the analysis, quantify peak deviations and beam-split probabilities, and \textcolor{black}{ demonstrate sum-rate gains obtained by including analytically predicted squinted/split and harmonic candidates in multi-RIS, multi-user, multi-band networks.}
\end{abstract}

\begin{IEEEkeywords}
Beam-split, beam-squint, harmonic decomposition, multi-band networks, peak characterization, phase quantization, Reconfigurable intelligent surface
\end{IEEEkeywords}

\section{Introduction}

Next-generation wireless networks are expected to operate across diverse transmission frequencies, either through wideband signaling around a designated carrier \cite{saeidi2024molecularabsorptionaware} or via simultaneous use of multiple carriers across distinct spectrum bands in multi-band networks (MBNs) \cite{MBN-Survey,MBN-Magazine,MBN-Mobility}. While these strategies enhance data rates and spectral flexibility, they also exacerbate frequency-dependent beamforming effects in large-aperture arrays and reconfigurable intelligent surfaces (RISs). When a beamforming profile, whether implemented at an active array or programmed at an RIS, is designed at a specific frequency, any deviation between the configuration and operating frequencies distorts the intended spatial response. This mismatch manifests as \textit{beam-squint}, where the main beam shifts from the desired direction, and \textit{beam-split}, where multiple \textcolor{black}{distinct maximum-gain} beams appear. These impairments intensify at upper mid-band, millimeter-wave, and terahertz frequencies, where large electrical apertures produce narrow beams that are inherently more sensitive to frequency variations. In highly directional transmission and reflection scenarios, such as RIS-assisted links, even minor frequency-induced angular deviations can translate into substantial array gain loss, undermining the expected benefits of wideband and multi-band operation. 
\begin{figure}
    \centering
    \includegraphics[width=1.03\linewidth]{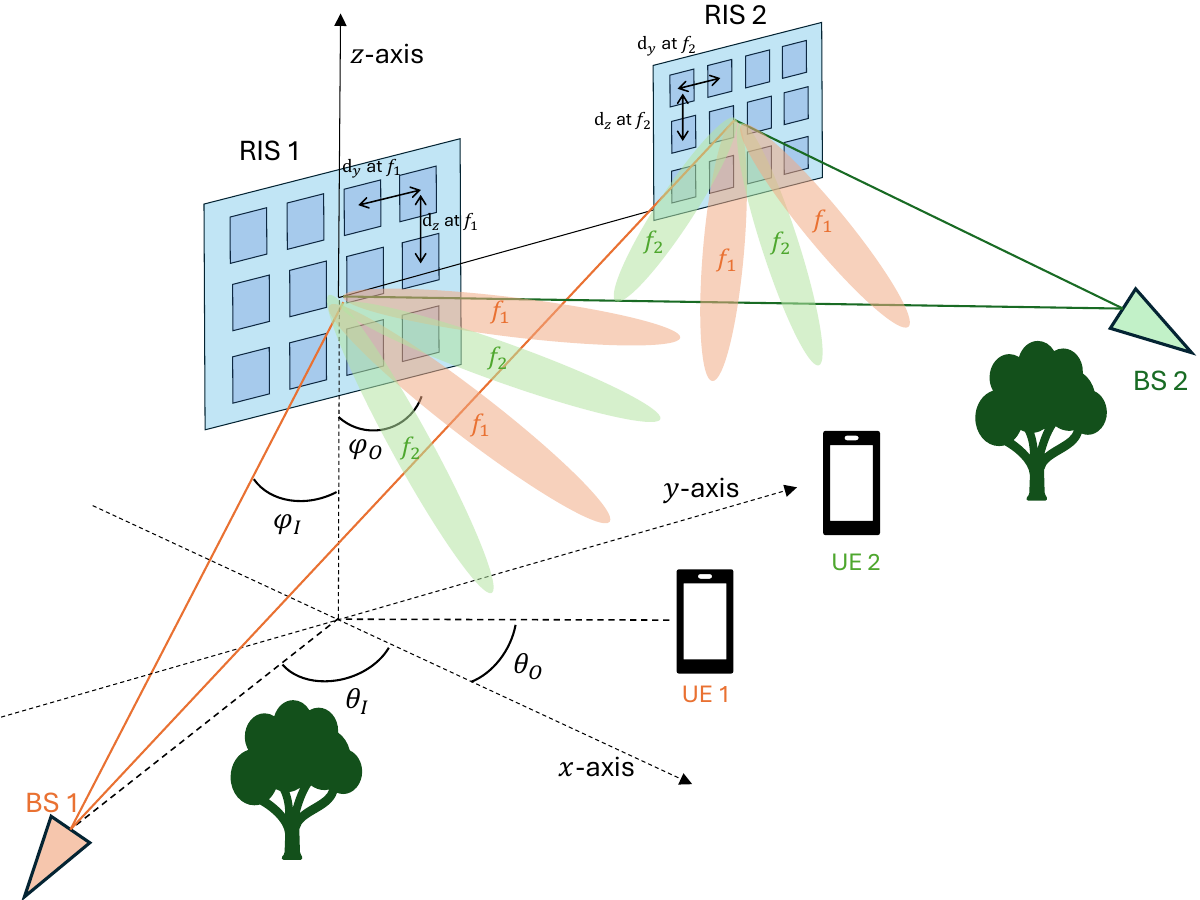}
    \caption{Two-RIS two-user multi-band network: RIS~1 (RIS~2) is configured at $f_1$ ($f_2$); colored beams illustrate beam-split under cross-band evaluation, with angles $(\varphi,\theta)$ defined by the shown coordinate axes.}

    \label{fig:SystemModel}
\end{figure}
\textcolor{black}{Consequently, multi-frequency systems require analytical tools capable of predicting how beam directions evolve with frequency. Such characterization is essential for frequency-aware beam management, wideband beamforming, and resource allocation, enabling efficient design without resorting to exhaustive angular scanning at every carrier.}

\subsection{Related Work}

A plethora of existing research works considered the problem of wideband beamforming in RIS-assisted networks. In this domain, the research studies compute beam-squint using a standard analytic expression assuming continuous RIS phase-profiles, and mitigate the beam-squint
via architectural mechanisms such as True-time-delay (TTD) structures \cite{su2023wideband,wang2024wideband,yan2024wideband, yan2023beamforming, sun2023beamforming, su2024joint,zhang2025near} and distributed deployment of RISs \cite{yan2023beamforming, sun2023beamforming, yashvanth2025mitigating}.

Beyond communication beamforming, beam-squint and beam-split analysis is also relevant to sensing and localization applications. The work in \cite{li2024user} leverages beam-squint and beam-split behavior in RIS-aided wideband mmWave sensing, where identifying frequency-dependent dominant directions is part of the sensing methodology, and delay structures are used to control the effect. Related integrated sensing and communication (ISAC) literature (not necessarily RIS-based) also highlights that squint-induced angle and range diversity can be exploited for sensing or localization \cite{luo2023beam,elbir2024antenna}. {More broadly, frequency-structured angular behavior can also arise from mechanisms beyond the intended RIS profile. For instance, \cite{kolomvakis2025nonlinear} obtained the angles and directions of nonlinear distortion components from large arrays/active RISs and used them for distortion-aware scheduling across frequency.}

Beam-squint effects are also relevant and typically incorporated into channel estimation by explicitly modeling frequency-dependent steering mismatch, since an angular dictionary built at a single carrier becomes mismatched across subcarriers and can bias angle recovery. For THz multiple-input multiple-output systems, \cite{elbir2023terahertz,elbir2023bsa} used frequency-dependent array responses to jointly recover sparse channel parameters and the associated subcarrier-dependent steering. For RIS-assisted THz systems, \cite{su2024channel,su2025two} redesigned channel recovery and training to account for subcarrier-dependent cascaded-channel steering, rather than assuming a frequency-invariant response. 
\textcolor{black}{
Optimization-based joint base station (BS)/RIS beamforming methods, e.g., \cite{jin2023low}, can also improve RIS-assisted transmission, but their goal is to optimize active/passive beamforming variables rather than characterize the frequency-dependent behavior of passive RISs.
} 

\subsection{Motivation and Contributions}

\textcolor{black}{Recent works have demonstrated the significance of frequency-dependent beam behavior analysis in various applications, such as wideband RIS beamforming \cite{su2023wideband,wang2024wideband,yan2024wideband, yan2023beamforming, sun2023beamforming, su2024joint,zhang2025near, yashvanth2025mitigating}, sensing/localization \cite{li2024user,luo2023beam,elbir2024antenna}, and channel estimation \cite{elbir2023terahertz,elbir2023bsa,su2024channel,su2025two}. However, these works often rely on mitigating or controlling these frequency-dependent beam misalignments through hardware or architectural modifications, such as TTD-based RISs or distributed RIS deployments. On the other hand, these studies predominantly consider continuous-phase RISs, whereas practical RIS implementations employ finite-resolution phase control \cite{bjornson2022reconfigurable}. This finite-resolution phase control introduces additional harmonic beams and fundamentally alters the frequency-dependent beam behavior. 
}

\textcolor{black}{This work addresses a complementary problem by establishing a theoretical characterization framework for \textit{conventional passive finite-resolution phase-only RISs}. Rather than mitigating beam squint through hardware modifications, the proposed framework derives closed-form predictions for the locations and existence conditions of squinted, split, and quantization-induced harmonic beams under frequency mismatch, thereby avoiding exhaustive two-dimensional angular searches. The angular locations of these beams can then be used for frequency-aware beam management in multi-frequency RIS operation, ISAC operation\footnote{\textcolor{black}{In ISAC operation, the predicted RIS-reflected maximum-gain directions can help reduce uncertainty in interpreting angular responses by identifying possible RIS-assisted communication/sensing directions and potential sources of angular ambiguity.}}, and to evaluate unintended leakage in multi-band or multi-operator scenarios.}

Specifically, in this paper, we mathematically formalize the beam-squint and beam-split considering uniform planar array (UPA) RISs. We define beam-squint as the deviation of the \textcolor{black}{maximum-gain} beam from the desired elevation–azimuth angle pair, and beam-split as the emergence of additional \textcolor{black}{distinct maximum-gain} beams when an RIS configuration designed at one frequency is evaluated at another. Specifically, we address the question of \textit{which elevation and azimuth angle pairs maximize the RIS beamforming gain under incident frequency mismatch, while considering both \textit{continuous and $b$-bit quantized phase} profiles?}

To this end, our contributions can be summarized as follows:

\noindent$\bullet$ We develop a unified mathematical framework to characterize the beam-squint and beam-split effects of UPA-RISs in the presence of both continuous phase and $b$-bit quantized phase profile.  For the quantization case, we propose a correction procedure to reduce the quantization-induced error in the beam-squint/split analytical estimates.

\noindent$\bullet$ For the continuous phase profile of the RIS, we derive necessary and sufficient conditions to maximize beamforming gain as well as feasibility conditions under which beam-squint and beam-split exist. We then obtain elevation--azimuth squint/split angle pairs in closed-form as a function of the incidence and configuration angles and frequencies, and element spacing of the UPA-RIS.

\noindent$\bullet$  For $b$-bit phase quantization, we derive the quantization-error phasor via its Fourier expansion, and transform the quantized UPA beampattern into an explicit superposition of harmonic array responses with affine phase slopes. This converts the non-affine quantized phase profile into tractable per-harmonic beamforming gain maximization problems. We then derive closed-form elevation--azimuth squint/split angle pairs, followed by deriving feasibility conditions, and a principled dominance ordering across harmonics.

\noindent$\bullet$  Since the per-harmonic beamforming gain maximizers (or peaks) can deviate from the true peaks of the \emph{superposed} beampattern, we develop a local perturbation analysis around each nominal harmonic peak and derive a closed-form deviation of the peak. This provides an analytically grounded refinement that improves peak-location characterization without exhaustive $2$D searches, especially for moderate RIS sizes.

\noindent $\bullet$ \textcolor{black}{We validate the derived squint/split expressions, feasibility conditions, and perturbation-based peak correction through Monte-Carlo simulations. We further show how the closed-form peak families can be used to construct compact beam-split-aware finite-codebook candidate sets in a multi-RIS, multi-user, multi-band network. The resulting codeword-selection method is evaluated under a line-of-sight (LoS)-dominant multipath channel and compared with home-link-only, squint-aware split-agnostic, and exhaustive finite-codebook search baselines.}

\textbf{Organization:} Section~II presents the system model and problem formulation. Section~III derives the elevation--azimuth squint/split angle pairs along with their feasibility conditions, considering the continuous-phase profile of the RIS. Section~IV extends the analysis to $b$-bit phase control via harmonic decomposition, and Section~V develops the perturbation-based peak correction under harmonic superposition. Section~VI depicts a beam-split-aware multi-RIS, multi-user, and multi-band network. Section~VII provides numerical validation, and Section~VIII concludes the paper.

\textbf{Notations:} Boldface lowercase/uppercase letters denote vectors/matrices. $\mathbb{R}$ and $\mathbb{C}$ denote the real and complex sets, respectively. $\abs{\cdot}$ is the absolute value and $\norm{\cdot}$ is the Euclidean norm; $\norm{\cdot}_F$ is the Frobenius norm. $(\cdot)^T$ and $(\cdot)^H$ denote transpose and conjugate transpose, $\Re\{\cdot\}$ the real part, and $(\cdot)^*$ complex conjugation. $\nabla(\cdot)$ denotes the gradient, and $\otimes$ denotes Kronecker products.

\section{System Model and Problem Statement}

Assume that a plane wave of frequency $f_I$ impinges on a UPA-based RIS from elevation angle $\varphi_I \in (-\frac{\pi}{2},\frac{\pi}{2})$ and azimuth angle $\theta_I \in (-\frac{\pi}{2},\frac{\pi}{2})$.
The array response vector can then be modeled to obtain the respective phase shifts among the RIS elements as follows:
\begin{align}
   \boldsymbol{a} &(\varphi_I,\theta_I,f_I,f_C) = [1,\dots,e^{j 2\pi(N_z-1)d_z\frac{ \sin\varphi_I}{\lambda_I}}]^T \notag \\ &
   \otimes [1,\dots,e^{j2\pi(N_y-1)d_y\frac{ \sin\theta_I\cos\varphi_I}{\lambda_I}}]^T
\end{align}
where $N_z$ and $N_y$ are the number of vertical and horizontal elements in the UPA, respectively, and $N = N_yN_z$ is the total number of elements. The wavelength of the incident signal is $\lambda_I = \frac{c}{f_I}$, with $c$ being the speed of light. $f_C$ is the design frequency of the RIS. The vertical element spacing is $d_z=\alpha_z \lambda_C = \frac{\alpha_z c}{f_C}$, and the horizontal element spacing is $d_y=\alpha_y \lambda_C = \frac{\alpha_y c}{f_C}$, where $\alpha_z > 0$ and $\alpha_y > 0$ are the factors controlling the element spacing. Hence,
the array response $\boldsymbol{a}(\varphi_I,\theta_I,f_I,f_C)$ depends on the frequency of the incident signal and the design frequency of the RIS. Although RIS elements may exhibit frequency-dependent phase responses \cite{abbas2024unit}, we assume frequency-independent unit-modulus elements in order to isolate the array-level beam-squint, beam-split, and quantization-induced harmonic effects. Consider the diagonal reflection matrix of the RIS as 
$\boldsymbol{\Psi}=\mathrm{diag}\left( e^{j\psi_1},\dots, e^{j\psi_N}\right)$, where $e^{j\psi_n}$ is the response of the $n$-th unit modulus element with the phase response denoted by $\psi_n \in \mathcal{F}_1= [0,2\pi)$ for the continuous case and $\psi_n \in \mathcal{F}_2= \{0,\Delta,\dots,(B-1)\Delta\}$ for the $b$-bit quantized RIS. $\Delta = \frac{2\pi}{B}$ denotes the quantization level and $B=2^b$.

The normalized beamforming gain of the RIS at an observation elevation angle of $\varphi_O$, and observation azimuth angle of $\theta_O$ is then given by \cite{su2023wideband}:
\begin{align}\label{eq:Arr-Gain}
    &u(\varphi_O,\theta_O,\boldsymbol{\Psi},\rho)  = \frac{1}{N}\abs{\boldsymbol{a}^T(\varphi_O,\theta_O,f_I,f_C)\boldsymbol{\Psi}\boldsymbol{a}(\varphi_I,\theta_I,f_I,f_C)} \notag \\ & = \frac{1}{N}\abs{\sum\limits_{n_y=0}^{N_y-1}\sum\limits_{n_z=0}^{N_z-1} e^{j\frac{2\pi}{\rho} (n_z\alpha_z\zeta_{I,O}+n_y\alpha_y\xi_{I,O})}  e^{j \psi_n}},
\end{align}
where $\rho = \frac{\lambda_I}{\lambda_C}=\frac{f_C}{f_I}$, $\zeta_{I,O} = \sin\varphi_I + \sin\varphi_O$, $\xi_{I,O} = \sin\theta_I\cos\varphi_I + \sin\theta_O\cos\varphi_O$, and $n = n_y + n_z N_y + 1$. \footnote{\textcolor{black}{The proposed analysis focuses on conventional single-layer passive UPA-RISs with diagonal phase responses. Extensions to stacked or beyond-diagonal RIS architectures require different effective aperture models and are left for future work.}}

In the following optimization problem, our objective is to identify all elevation and azimuth angle pairs $\{(\varphi_O^{*}, \theta_O^{*})\}$ at which the RIS beamforming gain attains its maximum. 
\begin{equation}\label{eq:SQ-Problem}
    \{(\varphi_O^{*}, \theta_O^{*})\} = \arg\max_{\varphi_O, \theta_O} \, u(\varphi_O, \theta_O, \boldsymbol{\Psi}, \rho),
\end{equation}
where $\boldsymbol{\Psi}$ is calculated by substituting $\rho = 1$ in the phases of the complex exponential given in \eqref{eq:Arr-Gain} and equating them to a constant phase $\phi_0$ as shown below, thus maximizing the beamforming  gain:
\begin{align}\label{eq:RIS-phases-rho-1} 
    2\pi(n_z\alpha_z\zeta_{I,D} + n_y\alpha_y\xi_{I,D}) + \psi_{n_y,n_z} \!\!\equiv \!\phi_{0,0},  \forall n_y,n_z, 
\end{align}
where $\phi_{0,0} \equiv \phi_0 \pmod{2\pi}$, which ensures that the phases of all complex exponential terms in \eqref{eq:Arr-Gain} are aligned to $\phi_0$. Also, $\zeta_{I,D}$ and $\xi_{I,D}$ are obtained based on $(\varphi_D, \theta_D)$, which denotes the design reflection angles of the RIS at frequency $f_C$. In particular, when $f_I = f_C$, i.e., $\rho = 1$, the maximum beamforming gain occurs at $(\varphi_D, \theta_D) \in \{(\varphi_O^*, \theta_O^*)\}$. \textcolor{black}{
In (3), the RIS phase profile is not an optimization variable and is fixed by the configuration rule in (4). The maximization is performed over the observation angles $(\phi_O,\theta_O)$ to identify the RIS-reflected maximum-gain beams generated by that fixed profile.
}
In the subsequent sections, we aim to solve the problem in \eqref{eq:SQ-Problem}, where $\boldsymbol{\Psi} \in \mathcal{F}_1$ corresponds to the continuous-phase setting and is obtained directly from \eqref{eq:RIS-phases-rho-1}, and $\boldsymbol{\Psi} \in \mathcal{F}_2$ corresponds to the quantized-phase setting, where $\boldsymbol{\Psi}$ is obtained by applying nearest-neighbor quantization to \eqref{eq:RIS-phases-rho-1}.

{\color{black}
Throughout this paper, the term ``dominant lobe'' is used in a maximum-gain sense. Dominant lobes are the feasible angle pairs that attain the maximum value of the normalized beamforming gain in \eqref{eq:SQ-Problem}. Therefore, a single maximum-gain angle pair, different from $(\varphi_D,\theta_D)$, corresponds to beam-squint, whereas two or more distinct maximum-gain angle pairs correspond to beam-split. 
}

\section{Beam-squint and Beam-split Existence and Characterization:  Continuous RIS Phase Profiles}

In this section, we first derive the necessary and sufficient conditions to maximize the beamforming gain. We then derive the squint/split angle pairs for the continuous RIS phase profile case, i.e., $\boldsymbol{\Psi} \in \mathcal{F}_1$, and analyze the beam-squint/split existence conditions for both elevation and azimuth directions. 

\subsection{Beam-squint and Beam-split Characterization}
We aim to identify all sets of observation angles $\{(\varphi_O^{*}, \theta_O^{*})\}$ for which the maximum beamforming gain is achieved when the RIS phase profiles are fixed continuous values and designed for $f_C$, but the frequency of the incident signal $f_I \neq f_C$, i.e., $\rho \neq 1$. The total phase at element $(n_y,n_z)$ becomes:
\begin{equation}\label{eq:Total-Phase-rhoNot1}
    \phi_{n_y,n_z}=\frac{2\pi}{\rho}(n_z\alpha_z \zeta_{I,O} + n_y\alpha_y \xi_{I,O}) + \psi_{n_y,n_z}.
\end{equation}
By plugging \eqref{eq:RIS-phases-rho-1} into \eqref{eq:Total-Phase-rhoNot1}, the total phase is obtained by:
\begin{align}\label{eq:Total-Phase-rhoNot1-RIS-at-rho1}
&\phi_{n_y, n_z}  =\frac{2 \pi}{\rho}\left(n_z \alpha_z \zeta_{I, O}+n_y \alpha_y \xi_{I, O}\right)+\phi_{0,0} \notag \\ 
& -2 \pi\left(n_z \alpha_z \zeta_{I, D}+n_y \alpha_y \xi_{I, D}\right)= \notag \\
 2 \pi&\left[n_z \alpha_z\left(\frac{1}{\rho} \zeta_{I, O}-\zeta_{I, D}\right)+n_y \alpha_y\left(\frac{1}{\rho} \xi_{I, O}-\xi_{I, D}\right)\right] \!+\! \phi_{0,0},
\end{align}
Next, to solve \eqref{eq:SQ-Problem} in this scenario, we first present the following Lemma. 
\begin{lemma}\label{lem:PhaseDiff}
    Using the triangle inequality, the beamforming gain can be bounded as follows: $$\tilde{u}=\abs{ \sum\limits_{n_y=0}^{N_y-1}\sum\limits_{n_z=0}^{N_z-1} e^{j\phi_{n_y,n_z}} } \leq  \sum\limits_{n_y=0}^{N_y-1}\sum\limits_{n_z=0}^{N_z-1} \abs{ e^{j\phi_{n_y,n_z}} } = N_yN_z,$$ where $\tilde{u} = N u$, and equality holds iff $\phi_{n_y+1,n_z}-\phi_{n_y,n_z}=2\pi m_y$ and $\phi_{n_y,n_z+1}-\phi_{n_y,n_z}=2\pi m_z$, for some integer values of $m_y,m_z \in \mathbb{Z}$; i.e., any difference between two neighboring phases must be an integer multiple of $2\pi$.
\end{lemma}

\begin{proof}
See \textbf{Appendix A}.
\end{proof}
{The equality conditions in Lemma~\ref{lem:PhaseDiff}} can be rewritten using \eqref{eq:Total-Phase-rhoNot1-RIS-at-rho1} as follows: 
\small
\begin{equation}\label{eq:z-peak-cond}
   m_z= \frac{\phi_{n_y,n_z+1}-\phi_{n_y,n_z}}{2\pi}=  \alpha_z\left(\frac{1}{\rho}\zeta_{I,O}-\zeta_{I,D}\right), \ m_z \in \mathbb{Z},
\end{equation}

\begin{equation}\label{eq:y-peak-cond}
    m_y = \frac{\phi_{n_y+1,n_z}-\phi_{n_y,n_z}}{2\pi}=\alpha_y\left(\frac{1}{\rho}\xi_{I,O}-\xi_{I,D}\right), \ m_y \in \mathbb{Z}.
\end{equation} \normalsize
Substituting $\zeta_{I,O} = \sin\varphi_I + \sin\varphi_O$ and $\zeta_{I,D} = \sin\varphi_I + \sin\varphi_D$ into \eqref{eq:z-peak-cond}, 
and then solving for $\varphi_O$, we obtain:
\begin{equation}\label{eq:varPhi-Sq}
    \varphi_O^{*} = \arcsin\left(\rho\left(\zeta_{I,D}+\frac{m_z}{\alpha_z}\right)-\sin\varphi_I\right),
\end{equation}

Moreover, substituting $\xi_{I, O}=\sin \theta_I \cos \varphi_I+\sin \theta_O \cos \varphi_O$ and $\xi_{I, D}=\sin \theta_I \cos \varphi_I+\sin \theta_D \cos \varphi_D$ into \eqref{eq:y-peak-cond}, and 
given $\varphi_O^*$ from \eqref{eq:varPhi-Sq} and defining $\cos \varphi_O^*=\sqrt{1-\sin ^2 \varphi_O^*} > 0$, $\theta_O^*$ is obtained as follows:
\begin{equation}\label{eq:Theta-Sq}
    \theta_O^* = \arcsin\left(\frac{\rho\left(\xi_{I,D}+\frac{m_y}{\alpha_y}\right)-\sin \theta_I \cos \varphi_I}{\sqrt{1-\left(\rho\left(\zeta_{I,D}+\frac{m_z}{\alpha_z}\right)-\sin\varphi_I\right)^2}}\right).
\end{equation}
Note that we can identify the \textit{beam-squint} by setting {$m_z = 0$ and $m_y = 0$ in \eqref{eq:varPhi-Sq} and \eqref{eq:Theta-Sq}, respectively}. \textcolor{black}{Other feasible integer pairs $(m_z,m_y)\neq(0,0)$ correspond to additional maximum-gain lobes of the same continuous-phase array factor and are therefore interpreted as beam-split lobes, not sidelobes.}


\subsection{Beam-Squint/Split Existence Conditions (Elevation Angles)}\label{Sec:Cont-Feas-Conditions}
It is important to identify the conditions under which beam-squint and beam-split occur, given the frequency ratio $\rho$ and the element spacing factors $\alpha_y$ and $\alpha_z$. {The feasibility conditions can be obtained by limiting the arguments of $\arcsin(\cdot)$ within the interval of $[-1,1]$.} Let $X_z$ denotes the argument of the $\arcsin(\cdot)$ function in \eqref{eq:varPhi-Sq}. By enforcing the condition $\abs{X_z} \leq 1$ and solving for $m_z$, we obtain the following integer set of feasible peaks, i.e.,
\begin{equation}\label{eq:Elev-Feas}
{\mathcal{M}_z} = \{m_z \in \mathbb{Z}: L_z \leq m_z \leq U_z\},
\end{equation}
where
\begin{equation}\label{eq:Feas-interv-LowB-z}
    L_z = \frac{\alpha_z}{\rho} \left( \sin\varphi_I - 1 - \rho \zeta_{I,D} \right),
\end{equation}
\begin{equation}\label{eq:Feas-interv-UpB-z}
    U_z = \frac{\alpha_z}{\rho} \left( \sin\varphi_I + 1 - \rho \zeta_{I,D} \right).
\end{equation}
Therefore, for a fixed pair $(\varphi_I, \varphi_D)$, the set of feasible elevation indices ($m_z$) consists of the integers contained within the real interval $[L_z, U_z]$. The length of the interval is given by
$U_z - L_z = \frac{2\alpha_z}{\rho},$
and its center is given as
$\frac{L_z + U_z}{2} = \frac{\alpha_z}{\rho} \left( \sin\varphi_I - \rho \left( \sin\varphi_I + \sin\varphi_D \right) \right).$ Based on this analysis, we establish the following  feasibility criteria:

\begin{itemize}
    \item If $\frac{2\alpha_z}{\rho} < 1$, the interval is shorter than 1, so for any fixed $(\varphi_I, \varphi_D)$, there can be \textit{at most} one integer $m_z$ in $[L_z, U_z]$ which represents beam-squint. However, \textit{no beam-split} occurs in the elevation direction. Moreover, depending on the center of the interval $[L_z, U_z]$, it is possible to have $\mathcal{M}_z = \emptyset$, and neither beam-squint nor beam-split can occur.

 \item  If $1 \leq \frac{2\alpha_z}{\rho} < 2$, the interval length is \textit{at least} 1 and less than 2, so for every $(\varphi_I, \varphi_D)$, there is at least one integer $m_z$ guaranteed in the interval which is referred to as beam-squint. On the other hand, {depending on the center, there may be two integers within the interval. In this case, both \textit{beam-split} and beam-squint occur.}

 \item  If $\frac{2\alpha_z}{\rho} \geq 2$: the interval always contains at least two integers, regardless of its center, so both  beam-squint and beam-split always occur.

\end{itemize}

We can determine the existence of beam-squint regardless of $\alpha_z$ and $(\varphi_I,\varphi_D)$ by
substituting $m_z = 0$ in \eqref{eq:varPhi-Sq}. The argument of the $\arcsin(\cdot)$ function is then given by $X_z^{(0)} = (\rho - 1)\sin\varphi_I + \rho \sin\varphi_D$. Applying the triangle inequality to $X_z^{(0)}$, we obtain:
\begin{equation}\label{eq:Elev-BeamSq-Cond}
\abs{X_z^{(0)}} \leq \abs{\rho - 1} \abs{\sin\varphi_I} + \rho \abs{\sin\varphi_D} \leq \abs{\rho - 1} + \rho.
\end{equation}

In \eqref{eq:Elev-BeamSq-Cond}, if $\rho < 1$, then $\abs{\rho - 1} + \rho = 1$. Consequently, for all $(\varphi_I, \varphi_D)$, we have $\abs{X_z^{(0)}} \leq 1$, and therefore $m_z = 0 \in \mathcal{M}_z$. 
On the other hand, for $\rho > 1$, beam-squint in elevation occurs ($m_z = 0 \in \mathcal{M}_z$) if the angle pair $(\varphi_I, \varphi_D)$ satisfies $\abs{X_z^{(0)}} \leq 1$.

\subsection{Beam-Squint/Split Existence Conditions (Azimuth Angles)}

To analyze the feasibility of beam-squint and beam-split in the azimuth direction, we assume that a feasible peak in the elevation direction exists for some $m_z$.
Let us define the numerator of the argument of $\arcsin(\cdot)$ in \eqref{eq:Theta-Sq} as:
$$
X_y = \rho\left(\sin \theta_I \cos \varphi_I + \sin \theta_D \cos \varphi_D + \frac{m_y}{\alpha_y} \right) - \sin \theta_I \cos \varphi_I,
$$
and let $A_O = \cos\varphi_O^* = \sqrt{1 - \sin^2\varphi_O^*}$, where $A_O \in (0,1]$.
Hence, the argument of the $\arcsin(\cdot)$ function in \eqref{eq:Theta-Sq} becomes $\frac{X_y}{A_O}$, and the azimuth angle $\theta_O^*$ is feasible if $\frac{\abs{X_y}}{A_O} \leq 1$.

Solving $\frac{\abs{X_y}}{A_O} \leq 1$ for $m_y$, we obtain the following integer set of feasible peaks, i.e.,
\begin{equation}\label{eq:Azim-Feas}
    \mathcal{M}_y = \{m_y \in \mathbb{Z}: L_y \leq m_y \leq U_y\},
\end{equation}
with
\begin{equation}\label{eq:Feas-interv-LowB-y}
    L_y = \frac{\alpha_y}{\rho} \left( \sin\theta_I\cos\varphi_I - \rho \xi_{I,D} - A_O \right),
\end{equation}
\begin{equation}\label{eq:Feas-interv-UpB-y}
    U_y = \frac{\alpha_y}{\rho} \left( \sin\theta_I\cos\varphi_I - \rho \xi_{I,D} + A_O \right).
\end{equation}

Therefore, for fixed pairs of $(\varphi_I, \varphi_D)$, $(\theta_I, \theta_D)$, and for a feasible $\varphi_O^*$, the set of feasible azimuth indices ($m_y$) consists of the integers within the real interval $[L_y, U_y]$. The length of this interval is $U_y - L_y = \frac{2\alpha_y}{\rho} A_O$, and its center is given by $
\frac{L_y + U_y}{2} = \frac{\alpha_y}{\rho} \left( \sin\theta_I\cos\varphi_I - \rho \xi_{I,D} \right)$.

Using $A_O \in (0,1]$ and the length of the interval, we establish the following  feasibility criteria:
\begin{itemize}
    \item If $\frac{2\alpha_y}{\rho}A_O < 1 \: \forall \varphi_O^*$, \textit{at most} one feasible solution exists which is referred to as beam-squint, and beam-split never occurs. Depending on the center of the interval $[L_y, U_y]$, it is possible to have $\mathcal{M}_y = \emptyset$, and therefore, neither beam-squint nor beam-split occurs.

    \item If $1 \leq \frac{2\alpha_y}{\rho}A_O < 2$, the length of the interval is \textit{at least} 1 and less than 2, so beam-split may occur depending on the center of the interval. {Depending on the center, there may be two integers within the interval. In this case, both \textit{beam-split} and beam-squint occur.}

    \item If $\frac{2\alpha_y}{\rho}A_O \geq 2$, the interval contains at least two integers, so both beam-squint and beam-split always occur.
\end{itemize}

 The feasibility of the main azimuth for $m_y = 0$ requires $\frac{\abs{X_y^{(0)}}}{A_O} \leq 1$, where $X_y^{(0)} = (\rho - 1)\sin\theta_I\cos\varphi_I + \rho \sin\theta_D\cos\varphi_D$. This condition implies that $\abs{X_y^{(0)}} \leq A_O = \cos\varphi_O^*$, which depends on all angles as well as $\rho$. Unlike in the elevation case, even when $\rho \leq 1$, azimuth squint is not guaranteed to occur, since $A_O$ can be strictly less than 1.

\begin{figure}
    \centering
    \includegraphics[width=0.95\linewidth]{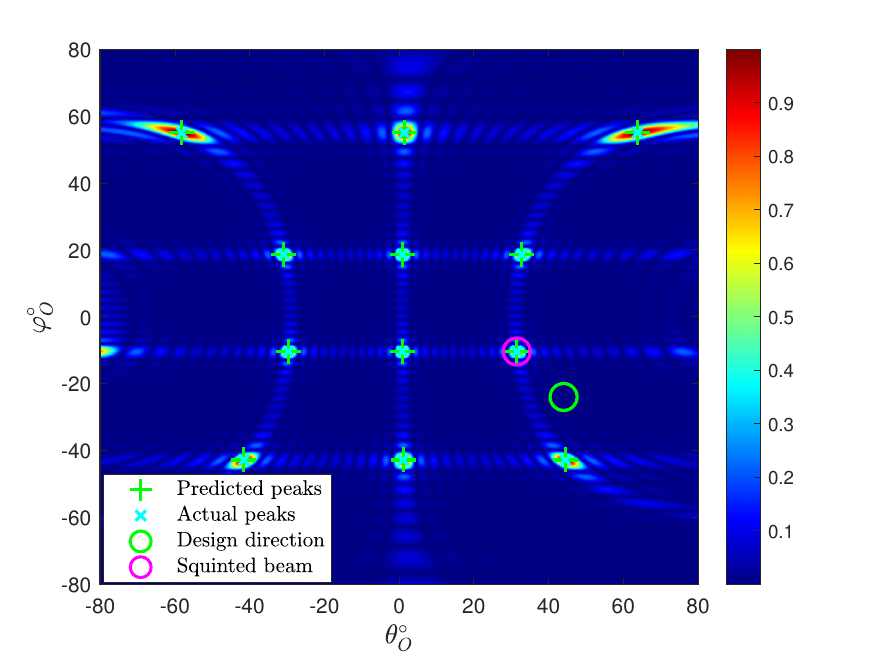}
    \caption{\textcolor{black}{Continuous RIS beampattern under frequency mismatch for $\rho=0.75$, $\alpha_y=\alpha_z=1.5$, $N_y=12$, $N_z=12$, with incidence $(\varphi_I,\theta_I)=(-30^\circ,-10^\circ)$ and design direction $(\varphi_D,\theta_D)=(-24^\circ,44^\circ)$. The green circle marks the design direction, while the magenta circle marks the squinted beam corresponding to $(m_z,m_y)=(0,0)$. The remaining maximum-gain peaks correspond to beam-split lobes.}}
\label{fig:TwoD_Corr_PerHar_Actu_LargeN_Many_BSP_Cont}
\end{figure}

\textbf{Remark:} {
For the continuous RIS phase profile, Fig.~\ref{fig:TwoD_Corr_PerHar_Actu_LargeN_Many_BSP_Cont} illustrates an example of the beam-squint and beam-split angle pairs obtained by the proposed framework, together with the actual peaks obtained by exhaustive search. This framework utilizes feasibility analysis based on $\mathcal{M}_z$ and $\mathcal{M}_y$ to limit the search space over the integer pairs $(m_z,m_y)$.
}

\section{Beam-squint and Beam-split Existence and Characterization: $b$-bit RIS Phase Profiles}
\label{sec:quant-harmonics}
In practice,  the phase shifts of the RIS elements are typically not continuous, thus cannot take arbitrary values from the set $\mathcal{F}_1$. Therefore, in this section, we  solve the problem in \eqref{eq:SQ-Problem} considering $b$-bit quantized RIS, i.e., for $\boldsymbol{\Psi} \in \mathcal{F}_2$.

In the continuous-phase case, the per-element phase is exactly affine in $(n_y,n_z)$, so the array sum achieves its maximum only when all unit-step phase increments are integer multiples of $2\pi$ (as given in \textbf{Lemma~\ref{lem:PhaseDiff}}). In contrast, under $b$-bit quantization, each element incurs an additional RIS phase shift error $\,\varepsilon(\psi_{n_y,n_z}) \in [-\Delta/2,\Delta/2)$, {which depends on the RIS element index $(n_y,n_z)$ and is no longer affine in $(n_y,n_z)$.} As a result, \textbf{Lemma~\ref{lem:PhaseDiff}} can no longer be utilized directly on the quantized array sum.

Subsequently, in this section, considering the Fourier series of the RIS quantization error phasor, we reformulate the beam pattern as a superposition of harmonic beam patterns and then obtain the closed-form beam-squint/split angle pairs that maximize each harmonic. We then analyze the dominance ordering of the harmonics and study the existence of per-harmonic beam-squint/split.

\subsection{Beam Pattern with $b$-bit Quantized RIS}
Considering $b$-bit nearest-neighbor quantization, we define the quantization function, which takes the continuous phase (with $\rho = 1$) from \eqref{eq:RIS-phases-rho-1} as input, as follows:
\begin{equation}\label{eq:quantizer}
    Q_b(\psi_{n_y,n_z}) = \Delta \left\lfloor \frac{\psi_{n_y,n_z}}{\Delta} \right\rceil,
\end{equation}
where $\lfloor x \rceil$ denotes the \textit{round} function.

We then define the {RIS quantization phase error function} as follows:
\begin{equation}\label{eq:phaseErrFunc}
    \varepsilon(\psi_{n_y,n_z}) \!= \!Q_b(\psi_{n_y,n_z}) - \psi_{n_y,n_z}\!\! =\! \Delta\! \left( \left\lfloor \frac{\psi_{n_y,n_z}}{\Delta} \right\rceil \!-\! \frac{\psi_{n_y,n_z}}{\Delta} \!\right)\!.
\end{equation}
Accordingly, the {beam pattern $S$} is obtained as:
\begin{align}\label{eq:BeamPattern_QuantErr}
    S &= \!\!\!\! \sum\limits_{n_y=0}^{N_y-1}\sum\limits_{n_z=0}^{N_z-1} e^{j\left( \frac{2\pi}{\rho}(\alpha_z\zeta_{I,O}n_z + \alpha_y\xi_{I,O}n_y) + Q_b(\psi_{n_y,n_z}) \right)} \notag \\
    &= \sum\limits_{n_y=0}^{N_y-1}\sum\limits_{n_z=0}^{N_z-1} e^{j\left( 2\pi(\delta_z n_z + \delta_y n_y) + \varepsilon(\psi_{n_y,n_z}) \right)},
\end{align}
where $\delta_z = \alpha_z\left(\frac{1}{\rho}\zeta_{I,O}-\zeta_{I,D}\right)$ and $\delta_y = \alpha_y\left(\frac{1}{\rho}\xi_{I,O}-\xi_{I,D}\right)$ are obtained by substituting $Q_b(\psi_{n_y, n_z})$ from \eqref{eq:quantizer} and performing some algebraic manipulations.

In particular, we define the RIS quantization error phasor by $g(x) = e^{j\varepsilon(x)}$, with $\varepsilon(x) = \Delta(\lfloor x\rceil - x)$, and for ease of exposition, we define $x$ by:
\begin{equation}\label{eq:Normalized-RIS-Phase}
x = \frac{\psi_{n_y,n_z}}{\Delta} = -B(n_z \alpha_z \zeta_{I,D} + n_y \alpha_y \xi_{I,D}) + x_0,
\end{equation}
where $x_0 = \frac{\phi_0}{\Delta}$. However, as we show below, the RIS quantization error phasor function  {has Fourier series representation}. Expanding this function into its harmonics decomposes the quantization effect into a weighted sum of array responses with \emph{purely linear} phases. For each harmonic, the phase is again affine in $(n_y,n_z)$, thus \textbf{Lemma~\ref{lem:PhaseDiff}} becomes applicable on a per-harmonic basis.

Since we aim to utilize the Fourier series of $g(x)$ in our analysis, it is necessary to show that the Fourier series of $g(x)$ exists and converges. This is given in \textbf{{Lemma-2}}.

\begin{lemma}\label{lem:FourierConve}
    The function $g: \mathbb{R}\rightarrow \mathbb{C}$, defined as $g(x) = e^{j\Delta(\lfloor x\rceil - x)}$, is {1-periodic and of bounded variation over one period}. Moreover, its Fourier series converges pointwise to $\frac{1}{2}(g(x^+)+g(x^-))$, where $g(x^-)$ and $g(x^+)$ denote the left- and right-hand limits of $g$ at $x$, respectively. 
\end{lemma}
\begin{proof}
See \textbf{Appendix B}.
\end{proof}
Next, we obtain the Fourier series of $g(x)$ in the following.
\begin{lemma}[Fourier series of $g$]\label{lem:FourierSeriesCoeff}
Let $g(x)=e^{j\,\Delta(\lfloor x\rceil-x)}$ with period $1$. {For harmonic $\ell\in\mathbb{Z}$, we define the Fourier series coefficients as follows:}
\[
\widehat g[\ell]=\int_{0}^{1} g(x)\,e^{-j2\pi \ell x}\,dx . \quad \mathrm{\textit{Then}}  \quad \widehat g[\ell]=\frac{2(-1)^{\ell}\,\sin\!\big(\tfrac{\Delta}{2}\big)}{\Delta+2\pi \ell}
\]
Consequently, the Fourier series can be given as:
\begin{equation}\label{eq:FourierCoeff-g}
    g(x)\sim \sum_{\ell\in\mathbb{Z}}\widehat g[\ell]\;e^{j2\pi \ell x},
\end{equation}
where $g(x)$ converges point-wise to $\tfrac12\big(g(x^-)+g(x^+)\big)$.
\end{lemma}
\begin{proof}
See \textbf{Appendix C}.
\end{proof}
\subsection{Beam Pattern Transformation and Per-harmonic Peaks}
At this point, we transform the beam pattern in \eqref{eq:BeamPattern_QuantErr} by substituting the Fourier series of the error phasor function $g(x)$, as given in \eqref{eq:FourierCoeff-g}, and obtain:
\begin{align}
    \widehat S & = \sum\limits_{n_y=0}^{N_y-1}\sum\limits_{n_z=0}^{N_z-1} e^{j\left( 2\pi(\delta_z n_z + \delta_y n_y) + \varepsilon(\psi_{n_y,n_z}) \right)} \notag \\ 
    & = \sum\limits_{\ell \in \mathbb{Z}} \widehat g[\ell] \sum\limits_{n_y=0}^{N_y-1}\sum\limits_{n_z=0}^{N_z-1} e^{j 2\pi(\delta_z n_z + \delta_y n_y)} e^{j 2\pi \ell x}.
\end{align}
where $x = \frac{\psi_{n_y,n_z}}{\Delta}$ is taken from \eqref{eq:Normalized-RIS-Phase}. Therefore, $\widehat S$ becomes:
\begin{align}\label{eq:beam-patt-err}
   \widehat S = \sum\limits_{\ell \in \mathbb{Z}} S_\ell =  \sum\limits_{\ell \in \mathbb{Z}} \widehat g[\ell] e^{j2\pi\ell x_0} \!\!\!\sum\limits_{n_y=0}^{N_y-1}\sum\limits_{n_z=0}^{N_z-1} \!\!\!e^{j 2\pi\left(\beta_z(\ell) n_z + \beta_y(\ell) n_y\right)},
\end{align}
where $\beta_z(\ell) = \delta_z-\ell B \alpha_z \zeta_{I,D}$ and $\beta_y(\ell) = \delta_y-\ell B \alpha_y \xi_{I,D}$.

Using \eqref{eq:beam-patt-err}, we first obtain the per-harmonic peak locations.
To do so, we define the absolute value of the $\ell$-th harmonic in \eqref{eq:beam-patt-err} as given below:
\begin{align}\label{eq:T_ell}
    T_{\ell}=\abs{\widehat g[\ell]} \abs{\sum\limits_{n_y=0}^{N_y-1}\sum\limits_{n_z=0}^{N_z-1}  e^{j 2\pi\left(\beta_z(\ell) n_z + \beta_y(\ell) n_y\right)}}.
\end{align}
Since the total phases are now affine in $(n_y,n_z)$ for a given $\ell$, by applying \textbf{Lemma}~\ref{lem:PhaseDiff}, the term $T_{\ell}$ is maximized if and only if the following set of equations holds:
\begin{equation}\label{eq:z-peak-cond-Quantization}
    \beta_z(\ell) = \delta_z - \ell B \alpha_z \zeta_{I,D} = m_z, \quad m_z \in \mathbb{Z},
\end{equation}
\begin{equation}\label{eq:y-peak-cond-Quantization}
    \beta_y(\ell) = \delta_y - \ell B \alpha_y \xi_{I,D} = m_y, \quad m_y \in \mathbb{Z}.
\end{equation}

For a fixed $\ell$, solving \eqref{eq:z-peak-cond-Quantization} for $\varphi_O$ yields:
\begin{equation}\label{eq:elev-squint-Quant}
    \varphi_O^*(\ell, m_z) \!=\! \arcsin\left(\rho\left(\zeta_{I,D}\! +\! \tfrac{1}{\alpha_z}(m_z \!+\! \ell B \alpha_z \zeta_{I,D})\right)\!-\!\sin\varphi_I\!\right)\!,
\end{equation}
and, given $\varphi_O^*(\ell, m_z)$, solving \eqref{eq:y-peak-cond-Quantization} for $\theta_O$ results in the optimal azimuth angle given at the top of the \textcolor{black}{next} page in \eqref{eq:azim-squint-Quant}.
\begin{figure*}
\begin{align}\label{eq:azim-squint-Quant} &\theta_O^*(\ell,m_z,m_y) = \arcsin\left(\frac{\rho\left(\xi_{I,D}+\tfrac{1}{\alpha_y}(m_y+\ell B\alpha_y \xi_{I,D})\right)-\sin\theta_I\cos\varphi_I}{\sqrt{1-\left(\rho\left(\zeta_{I,D}+\tfrac{1}{\alpha_z}(m_z+\ell B \alpha_z \zeta_{I,D})\right)-\sin\varphi_I\right)^2}}\right). 
\end{align}
\hrule
\end{figure*}
From \eqref{eq:elev-squint-Quant} and \eqref{eq:azim-squint-Quant}, we obtain the elevation and azimuth angles $\varphi_O^*(\ell, m_z)$ and $\theta_O^*(\ell, m_z, m_y)$ corresponding to the $\ell$-th harmonic for each integer pair $(m_z, m_y)$, respectively. 
{\color{black}
These angle pairs are maximum-gain lobes of the isolated $\ell$-th harmonic response. Specifically, $(\ell=0,m_z=0,m_y=0)$ corresponds to beam squint, whereas $(\ell\neq0,m_z\in \mathbb{Z},m_y\in \mathbb{Z})$ corresponds to harmonic-induced beams and is also counted as beam-split in the quantization case.
}
\subsection{Existence of Per Harmonic Beam-Squint/Split}

To analyze the existence of per-harmonic peaks, we first restrict attention to a finite set of ``dominant'' harmonics whose relative strength exceeds a threshold $\eta$, and then carry out the feasibility analysis only for $\ell$ in this set. Define the relative strength of the $\ell$-th harmonic relative to $\ell=0$ as:

\begin{equation}\label{eq:Relative-magnit}
    \kappa_{\ell} = \frac{\abs{\widehat g[\ell]}}{\abs{\widehat g[0]}} = \frac{\Delta}{\abs{\Delta + 2\pi \ell}} = \frac{1}{\abs{1 + B \ell}}.\footnote{For 1-bit quantization, we have $\kappa_0 = \kappa_{-1} > \kappa_\ell$ for all $\ell \in \mathbb{Z} \backslash \{-1, 0\}$, whereas for $b \geq 2$-bit quantization, $\kappa_0 > \kappa_\ell$ for all $\ell \in \mathbb{Z} \backslash \{0\}$.}
\end{equation}
Let $\eta \in (0,1]$ be a relative strength threshold. Then, the set of harmonics whose relative strength compared to $\ell = 0$ is at least $\eta$ is given by:
\begin{equation}\label{eq:Worthy-harmonics}
    \mathcal{K}_\eta = \left\{ \ell \in [-K_-,K_+] \cap \mathbb{Z} \right\},
\end{equation}
where $K_+ = \left\lfloor (-1+\tfrac{1}{\eta})B^{-1} \right\rfloor$ is obtained by solving $\kappa_\ell \geq \eta$ for $\ell \geq 0$, and $K_- = \left\lfloor (1+\tfrac{1}{\eta})B^{-1} \right\rfloor$ is obtained by solving $\kappa_\ell \geq \eta$ for $\ell < 0$. Here, $\lfloor \cdot \rfloor$ denotes the floor function.

Now, for all selected harmonics $\ell \in \mathcal{K}_\eta$, we can apply the feasibility analysis described in {Section~III}. In particular, the lower and upper bounds of the intervals containing feasible integers in the elevation direction, analogous to \eqref{eq:Feas-interv-LowB-z} and \eqref{eq:Feas-interv-UpB-z}, for the $\ell$-th harmonic are given by:
\begin{equation}\label{eq:Feas-interv-LowB-z-Quant}
    L_z(\ell) = \frac{\alpha_z}{\rho} \left( \sin\varphi_I - 1  - \rho \zeta_{I,D}\left(\ell B + 1\right) \right), \ \forall \ell \in \mathcal{K}_\eta
\end{equation}
\begin{equation}\label{eq:Feas-interv-UpB-z-Quant}
    U_z(\ell) = \frac{\alpha_z}{\rho} \left( \sin\varphi_I + 1  - \rho \zeta_{I,D}\left(\ell B + 1\right) \right), \ \forall \ell \in \mathcal{K}_\eta
\end{equation}
such that $L_z(0) = L_z$ and $U_z(0) = U_z$. In the $\ell$-th interval, an integer $m_z$ exists if $\lceil L_z(\ell) \rceil \leq \lfloor U_z(\ell) \rfloor$. If this condition is satisfied, then every integer $m_z \in \mathcal{M}_z(\ell) = [\lceil L_z(\ell) \rceil, \lfloor U_z(\ell) \rfloor]$ yields a feasible elevation angle $\varphi_O^*(\ell, m_z)$.

Assuming that $\varphi_O^*(\ell, m_z)$ exists such that $A_O(\ell) = \cos\varphi_O^*(\ell, m_z) > 0$, the lower and upper bounds of the interval containing feasible azimuth integers, similar to \eqref{eq:Feas-interv-LowB-y} and \eqref{eq:Feas-interv-UpB-y}, for the $\ell$-th harmonic are obtained by:
\begin{equation}\label{eq:Feas-interv-LowB-y-Quant}
    L_y(\ell) \!=\! \frac{\alpha_y}{\rho} \!\left(\! \sin\theta_I\cos\varphi_I \!-\! \rho \xi_{I,D}(\ell B+1) - A_O(\ell) \!\right)\!, \forall \ell \in \mathcal{K}_\eta
\end{equation}
\begin{equation}\label{eq:Feas-interv-UpB-y-Quant}
    U_y(\ell) \!=\! \frac{\alpha_y}{\rho} \!\left(\! \sin\theta_I\cos\varphi_I \!-\! \rho \xi_{I,D}(\ell B+1) + A_O(\ell) \!\right)\!, \forall \ell \in \mathcal{K}_\eta
\end{equation} 
with $L_y(0) = L_y$ and $U_y(0) = U_y$. Similarly, in the $\ell$-th interval, an integer $m_y$ exists if $\lceil L_y(\ell) \rceil \leq \lfloor U_y(\ell) \rfloor$. If this condition holds, then every integer $m_y \in \mathcal{M}_y(\ell) = [\lceil L_y(\ell) \rceil, \lfloor U_y(\ell) \rfloor]$ yields a feasible azimuth angle $\theta_O^*(\ell, m_z, m_y)$. Further analysis on the feasibility of beam-squint and beam-split of the per-harmonic peak locations can then be carried out, analogous to the continuous-phase case.

\section{Per-harmonic beam-squint/Split Angle Correction Mechanism}

From \eqref{eq:beam-patt-err}, it is evident that the quantized beam pattern $\widehat S$ is expressed as the superposition of infinitely many harmonics. In particular, {it can be observed in Fig.~\ref{fig:TwoD_Corr_PerHar_Actu} that} the superposition of multiple harmonics and their respective split beams can displace the actual per-harmonic peak positions resulting from the aggregate beam pattern $S$. To address this, we propose a correction method that refines the predicted peak locations by accounting for the interference introduced by other harmonics. Then, we analyze how the number of RIS elements and the frequency ratio can affect the accuracy of the quantized peak locations.

Considering the fact that multiplying a constant and mapping of $a \rightarrow a^2$ do not change the $\arg\max$ of the function in \eqref{eq:SQ-Problem}, we define the aggregate beam power pattern as follows:
\begin{equation}\label{eq:Total-Pow_pattermn}
    J(\varphi_O,\theta_O) = \abssSq{ \widehat S(\varphi_O,\theta_O)},
\end{equation}
and for a particular harmonic $\ell \in \mathbb{Z}$, define  $J_\ell(\varphi_O,\theta_O) = \abssSq{S_\ell(\varphi_O,\theta_O)}$. Moreover, for harmonic $\ell \in \mathcal{K}_\eta$ and integer pair $(m_z,m_y) \in \mathcal{M}_z(\ell) \times \mathcal{M}_y(\ell)$, we define
\begin{equation}    \boldsymbol{p}=\left[\begin{array}{c}\varphi_O^*\left(\ell, m_z\right) \\ \theta_O^*\left(\ell, m_z, m_y\right)\end{array}\right],
\end{equation} which is the main lobe center for this harmonic and integer pair. 
{\color{black}
At point $\boldsymbol{p}$, the \textit{isolated} harmonic power has a local maximum, i.e., $\nabla J_\ell(\boldsymbol{p})=\boldsymbol 0$. Moreover, for $N_y,N_z\geq2$, $\alpha_y,\alpha_z>0$, and $\varphi_O^*,\theta_O^* \in (-\tfrac{\pi}{2},\tfrac{\pi}{2})$, its Hessian satisfies $\boldsymbol H_{J_\ell}(\boldsymbol p)\prec0$, as shown in \textbf{Appendix~D}.
}

Now, we aim to estimate the actual maximum of  $J$ in a neighborhood of $\boldsymbol{p}$, which is further closer to the actual location of this specific lobe once all harmonics are superposed.
\subsection{Decomposition of Harmonic of Interest and Perturbation}
First, we decompose the {aggregate beam power pattern} into the main harmonic of interest at point $\boldsymbol{p}$, denoted by $J_\ell(\boldsymbol{p})$, and the perturbation, which is the contribution of all other harmonics, denoted by $R_\ell(\boldsymbol{p})$, i.e.,
\begin{align}\label{eq:Total-Pow-Patt-Extend}
    J(\boldsymbol{p}) &= J_\ell(\boldsymbol{p}) + R_\ell(\boldsymbol{p}) 
=
    \abssSq{\sum\limits_{\ell} S_{\ell}(\boldsymbol{p})}  \notag \\ 
    & = J_\ell(\boldsymbol{p})  + \underbrace{\sum\limits_{i\neq \ell} \abssSq{S_{i}(\boldsymbol{p})} + 2\sum\limits_{k < r} \Re\left\{S_{k}(\boldsymbol{p}) S^*_{r}(\boldsymbol{p})\right\} \hspace{-1mm}}_{R_\ell(\boldsymbol{p})}.
\end{align}
We assume that the perturbation $R_\ell(\boldsymbol{p})$ is locally smaller than $J_\ell(\boldsymbol{p})$, so we can approximate the gradient of $J$ using a first-order approximation as given in \textbf{Lemma~\ref{lem:1stApprx-VectorFunc}}.

\begin{lemma}[First-order approximation of a vector-valued function] \label{lem:1stApprx-VectorFunc}
Let $\boldsymbol{F}: \mathbb{R}^n \rightarrow \mathbb{R}^m$ be a vector-valued function that is differentiable at point $\boldsymbol{a} \in \mathbb{R}^n$, and let $\boldsymbol{F}^\prime(\boldsymbol{a}) \in \mathbb{R}^{m\times n}$ denote its Jacobian matrix at $\boldsymbol{a}$. Then, for any vector $\boldsymbol{h} \in \mathbb{R}^n$, we have:
$
\boldsymbol{F}(\boldsymbol{a} + \boldsymbol{h}) = \boldsymbol{F}(\boldsymbol{a}) + \boldsymbol{F}^\prime(\boldsymbol{a})\, \boldsymbol{h} + \boldsymbol{r}(\boldsymbol{h}),
$
where the remainder $\boldsymbol{r}(\boldsymbol{h}) \in \mathbb{R}^m$ satisfies
$\lim_{\boldsymbol{h}\to \boldsymbol{0}} \tfrac{\norm{\boldsymbol{r}(\boldsymbol{h})}}{\norm{\boldsymbol{h}}} = 0$
Hence, $\boldsymbol{F}$ admits the first-order approximation:
$\boldsymbol{F}(\boldsymbol{a} + \boldsymbol{h})\approx \boldsymbol{F}(\boldsymbol{a}) + \boldsymbol{F}^\prime(\boldsymbol{a}) \boldsymbol{h}$
with small error for small $\boldsymbol{h}$.
\end{lemma}

\begin{proof}
This follows directly from the definition of differentiability and Theorems 2.4 and 4.1 of \cite{edwards2012advanced}.
\end{proof}

Now, in order to estimate the actual peak location, denoted by $\bar{\boldsymbol{p}}$, of the main lobe for harmonic $\ell$ near $\boldsymbol{p}$, we define the shift vector as
$\Delta \boldsymbol{p} = \bar{\boldsymbol{p}} - \boldsymbol{p} = 
\begin{bmatrix}
\Delta \varphi \\
\Delta \theta
\end{bmatrix}.$
Note that $\bar{\boldsymbol{p}}$ is a local maximum of the total beam power pattern $J$ and therefore satisfies $\nabla J(\bar{\boldsymbol{p}}) = 0$.
Assuming that the actual peak $\bar{\boldsymbol{p}}$ lies close to the nominal peak $\boldsymbol{p}$ of harmonic $\ell$, we linearize the vector function $\nabla J(\bar{\boldsymbol{p}})$ around the point $\boldsymbol{p}$. 

By applying \textbf{Lemma}~\ref{lem:1stApprx-VectorFunc} to $\nabla J(\bar{\boldsymbol{p}})$, we have:
\begin{equation}\label{eq:1stApprx-on-TotalPatt-Grad}
    \nabla J(\boldsymbol{p} + \Delta \boldsymbol{p}) \approx \nabla J(\boldsymbol{p}) + \boldsymbol{H}_J(\boldsymbol{p}) \Delta\boldsymbol{p},
\end{equation}
where $\boldsymbol{H}_J$ is the Hessian of $J$, i.e., Jacobian matrix of $\nabla J$. Moreover, $\nabla J(\boldsymbol{p}) = \nabla R_\ell(\boldsymbol{p)}$ since $\nabla J_\ell(\boldsymbol{p}) = 0$. Solving \eqref{eq:1stApprx-on-TotalPatt-Grad} by using $\boldsymbol{H}_J(\boldsymbol{p})$ directly to find the angle shift $\Delta \boldsymbol{p}$ corresponds to a Newton step on the stationary condition $\nabla J(\bar{\boldsymbol{p}}) = 0$. However, since $J$ is highly nonconvex and contains many extrema, this step can move the estimate toward side-lobes or even null directions. 

\subsection{Local Surrogate Function-based Optimization}
Instead of using Newton step by using the Hessian of $J$, we maximize a local surrogate that preserves the dominant harmonic's negative-definite curvature while treating the residual $R_\ell$ as a linear perturbation. Since $\boldsymbol p$ is the isolated maximizer of $J_\ell$, we have $\nabla J_\ell(\boldsymbol p)=\boldsymbol 0$. We approximate 
\begin{equation}
J(\boldsymbol p+\Delta\boldsymbol p) \approx \tilde{J}(\Delta \boldsymbol{p}),
\end{equation}
where $J(\boldsymbol p+\Delta\boldsymbol p)=J_\ell(\boldsymbol p+\Delta\boldsymbol p)+R_\ell(\boldsymbol p+\Delta\boldsymbol p)$ and 
\begin{equation}\label{eq:surrgate_Delta}
    \tilde{J}(\Delta \boldsymbol{p}) \triangleq J_\ell(\boldsymbol{p})
    + \tfrac{1}{2}\Delta \boldsymbol{p}^T \boldsymbol{H}_{J_\ell}(\boldsymbol{p}) \Delta \boldsymbol{p}
    + R_\ell(\boldsymbol{p}) + \nabla R^T_\ell(\boldsymbol{p})\Delta \boldsymbol{p},
\end{equation}
which is second order in $J_\ell$ and first order in $R_\ell$ around $\boldsymbol{p}$. Since $\boldsymbol{H}_{J_\ell}(\boldsymbol{p}) \prec 0$, $\tilde{J}(\Delta \boldsymbol{p})$ is strictly concave in $\Delta\boldsymbol p$, and setting $\nabla_{\Delta \boldsymbol{p}} \tilde{J} = \boldsymbol 0$ yields:
\begin{equation}\label{eq:shift-Estimate}
    \Delta \boldsymbol{p} = - \boldsymbol{H}^{-1}_{J_\ell}(\boldsymbol{p}) \nabla R_\ell(\boldsymbol{p}).
\end{equation}
Thus, \eqref{eq:shift-Estimate} is the unique maximizer of this local surrogate, using the dominant harmonic's curvature rather than the full Hessian of $J$.
The expressions of the gradient and the Hessian used in \eqref{eq:shift-Estimate} are presented in \textbf{Appendix E}.

{\color{black}
The accuracy of the per-harmonic peak locations and the shift estimate in
\eqref{eq:shift-Estimate} depends on whether the residual term
$R_\ell$ remains a local perturbation relative to the isolated harmonic
lobe $J_\ell$, and whether the local curvature
$\boldsymbol H_{J_\ell}(\boldsymbol p)$ is sufficiently strong. From
\eqref{eq:shift-Estimate}, a large residual gradient
$\nabla R_\ell(\boldsymbol p)$ or a weak local curvature can produce a
large shift $\Delta\boldsymbol p$. Such a large shift may violate the
local assumption used in the first-order approximation of
$R_\ell(\boldsymbol p+\Delta\boldsymbol p)$ around $\boldsymbol p$.
To quantify this condition, we define the following validity metric:
\begin{equation}
\chi(\boldsymbol p)
=
\frac{\abssOne{
\nabla R_\ell^T(\boldsymbol p)\Delta \boldsymbol p}}{
J_\ell(\boldsymbol p)
},
\label{eq:chi-validity-1}
\end{equation}
where the numerator is the first-order change of the residual term over the predicted correction estimate, normalized by the harmonic peak power $J_\ell(\boldsymbol p).$ Hence,
$\chi(\boldsymbol p)<1$ indicates that the residual-induced local
variation over the predicted shift is smaller than the isolated harmonic
peak power, with smaller values corresponding to a more reliable local
approximation. Therefore,
$\chi$ provides a quantitative local-validity criterion for the
perturbation correction in \eqref{eq:shift-Estimate}. Moreover, using \eqref{eq:shift-Estimate}, we have
$\boldsymbol H_{J_\ell}(\boldsymbol p)\Delta\boldsymbol p
+\nabla R_\ell(\boldsymbol p)=\boldsymbol 0$, and the metric in
\eqref{eq:chi-validity-1} can be equivalently written as $\chi(\boldsymbol p)
=
\frac{
\Delta\boldsymbol p^T
[-\boldsymbol H_{J_\ell}(\boldsymbol p)]
\Delta\boldsymbol p
}{
J_\ell(\boldsymbol p)
},$ where the numerator is nonnegative since $\boldsymbol{H}_{J_\ell}(\boldsymbol{p}) \prec 0$, as shown in
{Appendix~D}. Additionally, from the local
second-order expansion of $J_\ell$ around $\boldsymbol p$, we have $
J_\ell(\boldsymbol p+\Delta\boldsymbol p)
\approx
J_\ell(\boldsymbol p) - \frac{1}{2}
\Delta\boldsymbol p^T
[-\boldsymbol H_{J_\ell}(\boldsymbol p)]
\Delta\boldsymbol p.$
Therefore, $\chi(\boldsymbol p)/2$ approximates the fractional drop of the
isolated harmonic lobe caused by the correction step. This shows that
$\chi(\boldsymbol p)$ measures whether the predicted correction remains
local in the lobe power sense, rather than only in terms of angular displacement.

Based on the negative-curvature result in {Appendix~D} and
the derivative expressions in {Appendix~E}, the dominant
terms of $\boldsymbol H_{J_\ell}(\boldsymbol p)$ scale as
$\abssSq{\widehat g[\ell]}\rho^{-2}N_y^2N_z^2
\left(\alpha_z^2(N_z^2-1)+\alpha_y^2(N_y^2-1)\right)$,
up to bounded angular factors and cross-curvature terms. Hence, increasing
$N_y$ and $N_z$ strengthens the local curvature of the isolated harmonic
lobe and reduces its sensitivity to the impact of residual harmonics, whereas
increasing $\rho$ weakens the curvature through the $1/\rho^2$ factor.
Since \eqref{eq:shift-Estimate} depends on
$\boldsymbol H_{J_\ell}^{-1}(\boldsymbol p)$, weaker curvature can produce
a larger correction step. Accordingly, in the numerical results, we use $\chi(\boldsymbol p)$ as the validity metric and also report
$\|\Delta\boldsymbol p\|_2$ only to show the angular size of the
correction.
}

\begin{figure*}[t]
\centering

\begin{minipage}[t]{0.44\textwidth}
    \centering
    \includegraphics[width=0.96\linewidth]{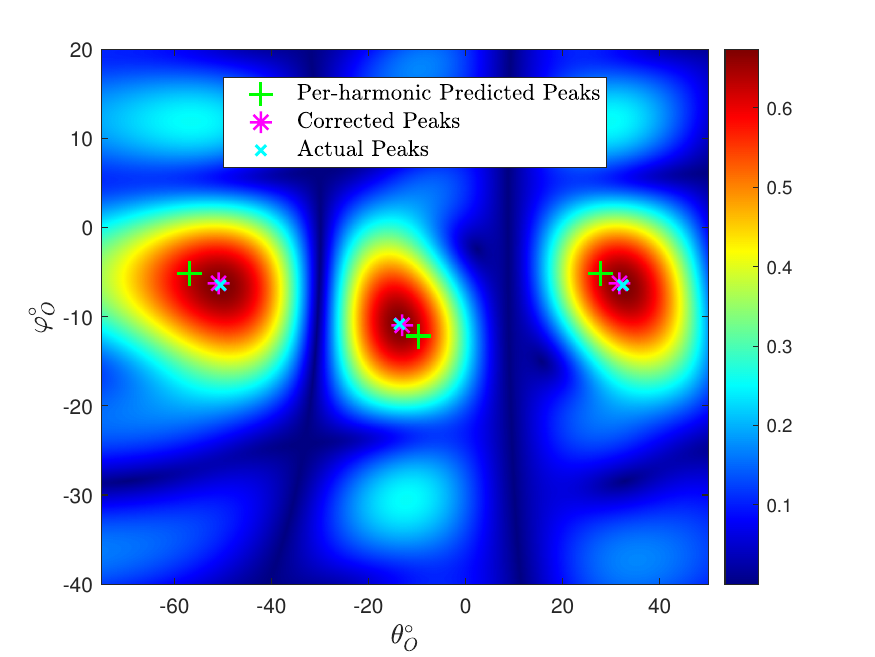}
    \caption{Quantized RIS beampattern under frequency mismatch for $\rho=0.65$, $\alpha_y=\alpha_z=0.5$, $N_y=4$, $N_z=6$, and $b=1$, with incidence $(\varphi_I,\theta_I)=(-30^\circ,-10^\circ)$ and design direction $(\varphi_D,\theta_D)=(-24^\circ,44^\circ)$.}
    \label{fig:TwoD_Corr_PerHar_Actu}
\end{minipage}\hfill
\begin{minipage}[t]{0.44\textwidth}
    \centering
    \includegraphics[width=0.96\linewidth]{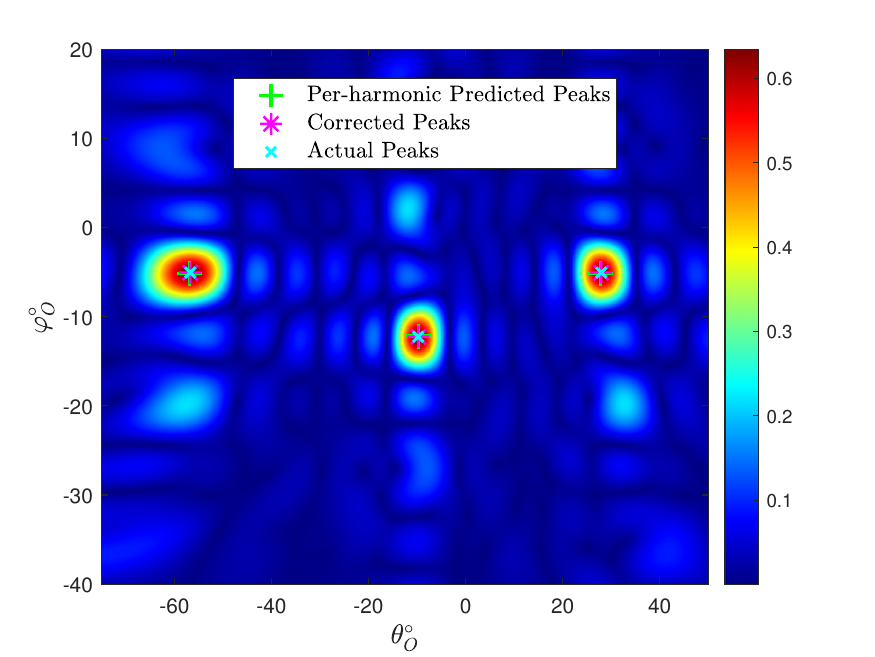}
    \caption{Quantized RIS beampattern under frequency mismatch for $\rho=0.65$, $\alpha_y=\alpha_z=0.5$, $N_y=12$, $N_z=16$, and $b=1$, with incidence $(\varphi_I,\theta_I)=(-30^\circ,-10^\circ)$ and design direction $(\varphi_D,\theta_D)=(-24^\circ,44^\circ)$.}
    \label{fig:TwoD_Corr_PerHar_Actu_LargeN}
\end{minipage}
\end{figure*}

\begin{figure}
    \centering
    \includegraphics[width=0.95\linewidth]{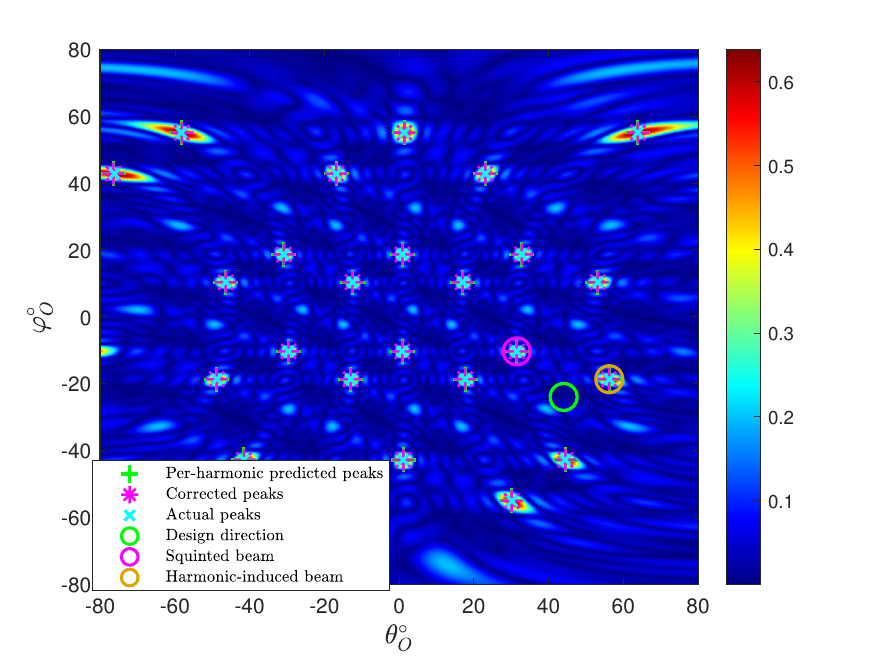}
    \caption{\textcolor{black}{Quantized RIS beampattern under frequency mismatch for $\rho=0.75$, $\alpha_y=\alpha_z=1.5$, $N_y=12$, $N_z=12$, and $b=1$, with incidence $(\varphi_I,\theta_I)=(-30^\circ,-10^\circ)$ and design direction $(\varphi_D,\theta_D)=(-24^\circ,44^\circ)$. The green circle marks the desired direction, the magenta circle marks the squinted beam associated with $(\ell,m_z,m_y)=(0,0,0)$, and the dark-yellow circle marks the quantization-induced harmonic beam associated with
$(\ell,m_z,m_y)=(-1,-3,2)$, which is considered as a beam-split lobe.}}
    \label{fig:TwoD_Corr_PerHar_Actu_LargeN_Many_BSP_Quant}
\end{figure}

{
Moreover, by comparing the beampatterns in Fig.~\ref{fig:TwoD_Corr_PerHar_Actu} and Fig.~\ref{fig:TwoD_Corr_PerHar_Actu_LargeN}, we note that increasing the number of RIS elements improves the accuracy of the per-harmonic peak locations, as well as that of the corrected peaks, by increasing the curvature of $J$ at the per-harmonic peak locations. Finally, Fig.~\ref{fig:TwoD_Corr_PerHar_Actu_LargeN_Many_BSP_Quant} illustrates the beampattern of a $1$-bit RIS with many beam-splits, all of which are detected by the proposed framework. Similarly to the continuous case, we utilize feasibility analysis to limit the search over the integer pairs $(m_z,m_y)$ for each harmonic $\ell \in \{-1,0\}$.
}

{
\subsection{Computational Complexity Analysis}
A direct beamforming gain peak search evaluates the beampattern (or $J=|S|^2$) over a 2D grid
$\mathcal G_\varphi\times\mathcal G_\theta$ with sizes $G_\varphi$ and $G_\theta$.
Since each evaluation of $S(\varphi,\theta)$ requires summing over all $N=N_yN_z$ RIS elements,
the exhaustive-scan complexity scales as $
\mathsf{Cost}_{\mathrm{sweep}}=O(G_\varphi G_\theta N),
$
which can be prohibitive for fine angular grids (large $G_\varphi$ and $G_\theta$) and large RIS sizes. In contrast, the complexity of our proposed framework is based on the following two factors.

\subsubsection{Per-harmonic peak complexity}
For each $\ell\in\mathcal K_\eta$, feasible peak indices satisfy
$m_z\in\mathcal M_z(\ell)=\mathbb Z\cap[L_z(\ell),U_z(\ell)]$ and
$m_y\in\mathcal M_y(\ell)=\mathbb Z\cap[L_y(\ell),U_y(\ell)]$.
Using $|\mathbb Z\cap[L,U]|\le (U-L)+2$ and the interval lengths
$U_z(\ell)-L_z(\ell)=\frac{2\alpha_z}{\rho}$ and
$U_y(\ell)-L_y(\ell)=\frac{2\alpha_y}{\rho}A_O(\ell)\le \frac{2\alpha_y}{\rho}$,
we obtain $|\mathcal M_z(\ell)|=O(\alpha_z/\rho+1)$ and $|\mathcal M_y(\ell)|=O(\alpha_y/\rho+1)$ $\forall  \ell \in \mathcal{K}_\eta$.
Hence, the total number of per-harmonic candidate peaks satisfies:
\begin{equation}
\sum_{\ell\in\mathcal K_\eta} |\mathcal M_z(\ell)|\,|\mathcal M_y(\ell)|
=
O\!\left(|\mathcal K_\eta|\left(\tfrac{\alpha_z}{\rho}+1\right)\left(\tfrac{\alpha_y}{\rho}+1\right)\right).
\end{equation}
Each candidate angle pair is obtained in constant time from \eqref{eq:elev-squint-Quant}--\eqref{eq:azim-squint-Quant}.
Therefore, the total complexity is obtained by: 
\begin{equation}
\mathsf{Cost}_{\mathrm{har}}
=
O\!\left(|\mathcal K_\eta|\left(\tfrac{\alpha_z}{\rho}+1\right)\left(\tfrac{\alpha_y}{\rho}+1\right)\right).
\end{equation}

\subsubsection{Perturbation-based correction complexity}
{\color{black}
The correction in \eqref{eq:shift-Estimate} is applied to each nominal per-harmonic candidate. The per-candidate cost depends on how the residual-gradient term is evaluated. If a truncated harmonic set $\mathcal K_c$ is used, each retained harmonic is separable along the $z$- and $y$-dimensions, so the field and first-order
derivatives cost $O(|\mathcal K_c|(N_y+N_z))$ per candidate. Hence,
\begin{equation}
\mathsf{Cost}_{\mathrm{corr}}^{\mathrm{trunc}}
=
O\!\left(
|\mathcal K_\eta|
\left(\frac{\alpha_z}{\rho}+1\right)
\left(\frac{\alpha_y}{\rho}+1\right)
|\mathcal K_c|(N_y+N_z)
\right).
\end{equation}
If the residual gradient is instead computed from the exact quantized field, the total field and its first-order derivatives require a sum over all $N_yN_z$ RIS elements at each candidate. The isolated-harmonic Hessian remains separable and is dominated by this exact gradient evaluation. Thus,
\begin{equation}
\mathsf{Cost}_{\mathrm{corr}}^{\mathrm{exact}}
=
O\!\left(
|\mathcal K_\eta|
\left(\frac{\alpha_z}{\rho}+1\right)
\left(\frac{\alpha_y}{\rho}+1\right)
N_yN_z
\right).
\end{equation}
}

\section{Use-Case Study: Beam-Split Aware Scheduling in Multi-RIS, Multi-User, Multi-band Network}
\label{subsec:usecase_twoRIS_twoUE_multiband}

In this section, to demonstrate the effectiveness of the obtained closed-form per-harmonic beam-squint/split angles in \eqref{eq:elev-squint-Quant}--\eqref{eq:azim-squint-Quant}, we consider a two-BS, two-RIS, two-UE downlink transmission scenario in an MBN, as illustrated in Fig.~\ref{fig:SystemModel}. User Equipment (UE)~1 is served on frequency band $f_1$ and UE~2 on frequency band $f_2$, with $
\rho_{12}= \frac{f_1}{f_2}$, $\rho_{21}=\frac{f_2}{f_1}$, and $\lambda_k=\frac{c}{f_k},\ k\in\{1,2\}$.
RIS~1 is configured at $f_1$ and RIS~2 at $f_2$. UE 1 is considered as the \textit{home} UE of RIS 1 and the \textit{cross} UE of RIS 2, and vice versa. 
\textcolor{black}{
Specifically, we assume that the direct BS--UE paths are blocked, which allows us to isolate the RIS-assisted frequency-mismatch effects in the considered use case. However, the presence of direct BS--UE links does not change the proposed characterization of the RIS-reflected maximum-gain beams.}  
\textcolor{black}{Moreover, we evaluate the achievable rate under a carrier-dependent LoS-dominant multipath channel. The proposed analytical framework is used to generate candidate RIS phase profiles only from the LoS geometry, while the final achievable rate is computed using the complete multipath cascaded channel. This separation is important: the closed-form split/squint peak families provide a low-complexity candidate-generation mechanism, whereas the rate evaluation accounts for the actual multipath channel realization.}
{\color{black}
Let $\boldsymbol h^{\mathrm{BR}}_{k,r}(f_k)\in\mathbb C^{N}$ denote the BS~$k$ to RIS~$r$ channel vector at carrier $f_k$, and let $\boldsymbol h^{\mathrm{RU}}_{r,i}(f_k)\in\mathbb C^{N}$ denote the RIS~$r$ to UE~$i$ channel vector at the same carrier. For a given RIS phase matrix $\boldsymbol{\Psi}_r$, the cascaded gain through RIS~$r$ is $
h_{k,r,i}(\boldsymbol{\Psi}_r;f_k)
=
\left(\boldsymbol h^{\mathrm{RU}}_{r,i}(f_k)\right)^T
\boldsymbol{\Psi}_r
\boldsymbol h^{\mathrm{BR}}_{k,r}(f_k).$
}
{\color{black}
Each channel is modeled as \cite{su2025two}:
\begin{equation}
\mathbf h_{a,b}(f)
=
\sqrt{\beta_{a,b}(f)}
\sum_{l=0}^{L_{a,b}}
\gamma_{a,b,l}
e^{-j2\pi f\tau_{a,b,l}}
\mathbf a(f,\varphi_{a,b,l},\theta_{a,b,l}),
\end{equation}
where $l=0$ denotes the LoS path and
$l=1,\ldots,L_{a,b}$ denote the NLoS paths. $\beta_{a,b}(f)$ is the large-scale path gain, $\gamma_{a,b,l}$ is complex gain, $\tau_{a,b,l}$ is the path delay, and $(\varphi_{a,b,l},\theta_{a,b,l})$ is the path direction at the RIS.
}
{\color{black}
Then, considering the orthogonality of the different frequency bands, the sum-rate is computed as follows:
\begin{align}\label{eq:Sum-Rate}
R
&=
\log_2\!\left(
1+
\frac{
P_1
\left|
h_{1,1,1}(\boldsymbol{\Psi}_1;f_1)
+
h_{1,2,1}(\boldsymbol{\Psi}_2;f_1)
\right|^2
}{\sigma_1^2}
\right)
\nonumber\\
&\quad+
\log_2\!\left(
1+
\frac{
P_2
\left|
h_{2,2,2}(\boldsymbol{\Psi}_2;f_2)
+
h_{2,1,2}(\boldsymbol{\Psi}_1;f_2)
\right|^2
}{\sigma_2^2}
\right),
\end{align}
}where $P_1$ and $P_2$ denote the transmit powers at BS 1 and BS 2, respectively, and $\sigma^2_1$ and $\sigma^2_2$ are the additive noise powers at UE 1 and UE 2, respectively. \textcolor{black}{Our objective is to maximize the sum-rate while allowing each RIS to provide cross-band assistance when a mismatched-band reflected beam is aligned with the corresponding cross-band UE.} In particular, the sum-rate depends on $\abssSq{h_1(\boldsymbol{\Psi}_1,\boldsymbol{\Psi}_2;f_1)}=\abssSq{h_{1,1,1}(\boldsymbol{\Psi}_1;f_1)+h_{1,2,1}(\boldsymbol{\Psi}_2;f_1)}$ for UE 1 and $\abssSq{h_2(\boldsymbol{\Psi}_1,\boldsymbol{\Psi}_2;f_2)}=\abssSq{h_{2,2,2}(\boldsymbol{\Psi}_2;f_2)+h_{2,1,2}(\boldsymbol{\Psi}_1;f_2)}$ for UE 2. Therefore, improving the \emph{cross} terms $h_{1,2,1}(\boldsymbol{\Psi}_2;f_1)$ and $h_{2,1,2}(\boldsymbol{\Psi}_1;f_2)$ (mismatched-band contribution) can increase the received power through both RISs.
}

{While \eqref{eq:RIS-phases-rho-1} can be instantiated for any continuous steering direction $(\varphi_D,\theta_D)$, practical beam management relies on selecting a profile from a finite RIS codebook to enable codeword-based beam training with limited overhead \cite{lv2024ris}.}
We therefore discretize the direction-cosine domain into a grid of resolution $Q_{\min}$, which controls the tradeoff between angular snapping error and the number of codewords evaluated. This direction discretization is distinct from the $b$-bit element-level phase quantization in \eqref{eq:quantizer}. 

In what follows, we describe the design of a beam-split-aware RIS phase-shift codebook with a given resolution $Q_{\min}$ and the corresponding codeword selection. We then present the associated simulation results in the next section. {Moreover, for ease of exposition, we adopt generic notation for the codebook design and codeword selection method for a given RIS.
}

\subsection{Quantized codebooks at $\rho=1$ (home and cross incidence)}
For each RIS, we construct a finite set of $b$-bit phase profiles $\{\boldsymbol{\Psi}(q)\}$ at $\rho=1$ by instantiating the continuous-phase design in \eqref{eq:RIS-phases-rho-1} for a grid of \emph{design} directions $(\varphi_D,\theta_D)$ and then applying the elementwise quantizer in \eqref{eq:quantizer}. To do so, we first parameterize directions using the direction-cosines implied by \eqref{eq:Arr-Gain},
$s_{z,D}=\sin\varphi_D$, $ s_{y,D}=\cos\varphi_D\sin\theta_D$,
and sample $(s_{z,D},s_{y,D})$ on a $Q_{\min}\times Q_{\min}$ uniform lattice, where $Q_{\min} \in \mathbb{Z}_+$ determines the resolution of sampling. 

Hence, for $|s_{z,D}|\le 1$ and $|s_{y,D}|\le\sqrt{1-s_{z,D}^2}$, we have
$\varphi_D=\arcsin(s_{z,D})$ and 
$\theta_D=\arcsin\!\left(\frac{s_{y,D}}{\cos\varphi_D}\right).$
For each $(\varphi_D,\theta_D)$, we plug $(\varphi_I,\theta_I)$ and $(\varphi_D,\theta_D)$ into the phase-alignment rule \eqref{eq:RIS-phases-rho-1}, quantize using \eqref{eq:quantizer}, and form the unit-modulus profile $\boldsymbol{\Psi}(q)=\exp(j\boldsymbol{\psi}_q)$, which corresponds to the $q$-th codeword.

We then precompute \emph{two} codebooks per RIS: a \emph{home} codebook built using the home BS$\to$RIS incidence angles, and a \emph{cross} codebook built using the cross BS$\to$RIS incidence angles. This separation is necessary because the incidence angles appear inside the phase-slope terms in \eqref{eq:RIS-phases-rho-1}, so changing the BS incidence changes the entire quantized phase pattern associated with a given $(\varphi_D,\theta_D)$. As a result, the home and cross-incidence cases generally induce different discrete mappings $(\varphi_D,\theta_D)\mapsto \boldsymbol{\Psi}(q)$, and both are required by the two-RIS two-BS geometry.

\subsection{Beam-split-aware Codeword Selection}
For each RIS, the cross-assistance goal is to find a \emph{small} set of quantized phase profiles (designed at $\rho=1$ using the cross-incidence codebook) whose \emph{mismatched} dominant lobe aligns with the cross-user direction at the other band. We avoid exhaustive $(\varphi,\theta)$ sweeps at $\rho\neq 1$ by using the per-harmonic peak-family structure in \eqref{eq:elev-squint-Quant}--\eqref{eq:azim-squint-Quant}.
Consider one RIS with known incidence angles $(\varphi_I,\theta_I)$ for the cross link (other BS$\to$RIS). Let $(\varphi_O,\theta_O)$ denote the \emph{desired} cross-user direction (RIS$\to$cross UE). For each tuple $(\ell,m_z,m_y)$ with $\ell \in \mathcal{K}_\eta$, $m_z \in \mathcal{M}_z(\ell)$ and $m_y \in \mathcal{M}_y(\ell)$, we compute a \emph{design} direction $(\varphi_D,\theta_D)$ at $\rho=1$ such that the $\ell$-th harmonic peak at mismatch $\rho$ coincides with $(\varphi_O,\theta_O)$.
Then, assuming $\rho = \rho_{12}$ for RIS 1 evaluated at $f_2$ (or $\rho = \rho_{21}$ for RIS 2      evaluated at $f_1$), using the peak-family equations used to obtain \eqref{eq:elev-squint-Quant}--\eqref{eq:azim-squint-Quant} and solving for $(\varphi_D,\theta_D)$ gives
\begin{align}
\sin\varphi_D(\ell,m_z)
&=
\frac{1}{A_\ell}\Big(\frac{\zeta_{I,O}}{\rho}-\frac{m_z}{\alpha_z}\Big)-\sin\varphi_I,
\label{eq:uc_inv_phi}\\
\sin\theta_D(\ell,m_z,m_y)
&=\!
\frac{1}{\cos\varphi_D}\!\!\left(
\!\frac{1}{A_\ell}\Big(\frac{\xi_{I,O}}{\rho}\!-\!\frac{m_y}{\alpha_y}\Big)
\!\!-\!\sin\theta_I\cos\varphi_I \!\!
\right)\!\!,
\label{eq:uc_inv_th}
\end{align}
where $A_\ell = 1+\ell B$, $\varphi_D=\arcsin(\sin\varphi_D)$ and $\theta_D=\arcsin(\sin\theta_D)$ are accepted only when $|\sin\varphi_D|\le 1$, $\cos\varphi_D>0$, and $|\sin\theta_D|\le 1$ hold. 
Because the RIS can only realize codebook profiles, each feasible $(\varphi_D,\theta_D)$ is mapped to the nearest codeword index in the precomputed \emph{cross} codebook (built at $\rho=1$ with incidence $(\varphi_I,\theta_I)$). Using the direction cosines already used for codebook construction, i.e., $s_{z,D}$ and $s_{y,D}$, we select
\begin{equation}
q=\arg\min_{q'}\Big[(s_{z,D}-s_z(q'))^2+(s_{y,D}-s_y(q'))^2\Big],
\end{equation}
where $(s_z(q'),s_y(q'))$ are the stored direction cosines of codeword $q'$ in the cross codebook. 
Since multiple tuples $(\ell,m_z,m_y)$ can map to the same $q$, we keep only one representative tuple per $q$ by re-evaluating the \emph{mismatched} peak that the selected codeword produces. More specifically, codeword $q$ at design angles $(\varphi_D(q),\theta_D(q))$ produces a peak at $(\varphi_P,\theta_P)$ at mismatch $\rho$. We then measure the residual misalignment to the desired cross-user direction $(\varphi_O,\theta_O)$ via $d_{\mathrm{cross}}(q,\ell,m_z,m_y)\triangleq
d\big((\varphi_P,\theta_P),(\varphi_O,\theta_O)\big)$, where $d(\cdot,\cdot)$ is the angular distance. Next, we must prune to a \emph{small} cross-aware candidate set of size $L_{\mathrm{keep}}$ so that the final two-RIS joint selection remains tractable in terms of time complexity. If we rank candidates only by $d_{\mathrm{cross}}^\star(q)\triangleq\min_{\ell,m_z,m_y} d_{\mathrm{cross}}(q,\ell,m_z,m_y)$, then for moderate $Q_{\min}$ and small $L_{\mathrm{keep}}$ the shortlist can become dominated by profiles that are excellent for the cross peak but correspond to design directions far from the home steering direction. To bias the shortlist toward candidates that remain compatible with home steering while still enabling cross assistance, we rank codewords using the following weighted angular distance:
\begin{equation}
\mathcal{L}(q)\triangleq d_{\mathrm{cross}}^\star(q)+w_{\mathrm{home}}\,d_{\mathrm{home}}(q),
\end{equation}
where $d_{\mathrm{home}}(q)\triangleq d\big((\varphi_D(q),\theta_D(q)),(\varphi_{O,\mathrm{home}},\theta_{O,\mathrm{home}})\big)$, and $w_{\mathrm{home}} \in [0,1]$ is the home-priority weight. We then
keep the $L_{\mathrm{keep}}$ smallest $\mathcal{L}(q)$ codewords as the cross-aware candidate set. Finally, we also include the home-maximizing codeword ($q_{\mathrm{home}}$) from the agnostic selection, which only considers the home UE per RIS, as a robust fallback.
Algorithm~\ref{alg:uc_bsaware_candidates} summarizes the per-RIS beam-split-aware candidate generation and pruning used to build a small cross-assistance shortlist without exhaustive angular sweeps at $\rho\neq 1$. We run Algorithm~\ref{alg:uc_bsaware_candidates} independently at RIS~1 and RIS~2 (with the appropriate $\rho$ for cross-band evaluation). Then, we select the operating pair that maximizes the sum-rate in \eqref{eq:Sum-Rate} over the candidate sets of codewords generated in Algorithm~\ref{alg:uc_bsaware_candidates}. \textcolor{black}{It is worth noting that Algorithm~1 does not impose constructive-interference constraints, and it uses the closed-form peak characterization to form a compact set of candidate codewords that can improve the cross-band cascaded gains. Extending the proposed framework toward constructive-interference-aware designs is an interesting direction for future work.}

\begin{algorithm}[t]
\caption{Beam-split-aware Codeword Selection in Multi-band RIS-assisted Networks}
\label{alg:uc_bsaware_candidates}
\begin{algorithmic}[1]
\State Define $\zeta_{I,O}$ and $\xi_{I,O}$ as in \eqref{eq:Arr-Gain} using $O=(\varphi_O,\theta_O)$.
\State Initialize $d^\star_{\mathrm{cross}}(q)\leftarrow +\infty$ for all $q$.
\For{$\ell \in \mathcal{K}_\eta$, and $(m_z,m_y) \in \mathcal{M}_z(\ell) \times \mathcal{M}_y(\ell)$}
  \State Compute $(\varphi_D,\theta_D)$ via \eqref{eq:uc_inv_phi}--\eqref{eq:uc_inv_th}.
  \State $(s_{z,D},s_{y,D})\leftarrow(\sin\varphi_D,\ \cos\varphi_D\sin\theta_D)$ and snap
  $q \leftarrow \arg\min_{q'} (s_{z,D}-s_z(q'))^2+(s_{y,D}-s_y(q'))^2$.
  \State Predict $(\varphi_P,\theta_P)$ for codeword $q$ at mismatch $\rho$ using \eqref{eq:elev-squint-Quant}--\eqref{eq:azim-squint-Quant}.
  \State $d^\star_{\mathrm{cross}}(q)\leftarrow \min\!\Big\{d^\star_{\mathrm{cross}}(q),\ d\big((\varphi_P,\theta_P),(\varphi_O,\theta_O)\big)\Big\}$.
\EndFor
\State $d_{\mathrm{home}}(q)\leftarrow d\big((\varphi_D(q),\theta_D(q)),(\varphi_{O,\mathrm{home}},\theta_{O,\mathrm{home}})\big)$.
\State $\mathcal{L}(q)\leftarrow d^\star_{\mathrm{cross}}(q) + w_{\mathrm{home}}\,d_{\mathrm{home}}(q)$ and
$\mathcal{Q}\leftarrow \text{indices of the } L_{\mathrm{keep}}\text{ smallest } \mathcal{L}(q)$.
\State $\mathcal{Q}\leftarrow \mathcal{Q}\cup\{q_{\mathrm{home}}\}$.
\end{algorithmic}
\end{algorithm}

\section{Numerical Results and Discussions}

In this section, we validate the closed-form expressions and the proposed correction by comparing predicted peak angles with exhaustive $2$D angular search, and by quantifying split probabilities and rate gains in the MBN use case. Unless stated otherwise, angles are drawn uniformly in the direction-cosine domain over $\varphi,\theta\in[-80^\circ,80^\circ]$ (to avoid degenerate cases), and Monte-Carlo averaging is performed with $5\times10^3$ trials per point. For quantized results, nearest-neighbor $b$-bit phase quantization is applied, and only the dominant harmonics in $\mathcal{K}_\eta$ (\eqref{eq:Worthy-harmonics} with $\eta > \tfrac{1}{3}$) are used for peak characterization. {Moreover, in the benchmarks, ``Cont" refers to the continuous RIS phase profile case, and ``$b$-bit no corr" and ``$b$-bit corr" refer to RIS phase profiles with $b$-bit quantization, without and with applying the proposed correction method, respectively.
}

\subsection{Beam-Squint/Split Analytical Validation}

Fig.~\ref{fig:err-Vs-rho} shows the mean angular error between the predicted peak locations and the true peaks obtained by exhaustive search as a function of the frequency ratio $\rho$, for two element spacings $\alpha=\alpha_z=\alpha_y\in\{0.5,0.85\}$. In the continuous-phase case, the error is essentially zero for all $\rho$ because the RIS array sum has an affine phase progression across the UPA, and the derived integer-slope peak conditions are necessary and sufficient; thus, the closed-form peak families coincide with the maximizers found by the exhaustive search.

Under $1$-bit quantization, exact phase affinity is lost and the beampattern becomes a superposition of harmonic array responses. Consequently, the maximizer of the total pattern can deviate from the per-harmonic predicted peak locations, yielding a non-zero error that generally increases with $\rho$. This trend is consistent with the reduction of dominant-harmonic curvature as $\rho$ grows, which increases sensitivity to interference from the remaining harmonics. The proposed correction alleviates this by applying a perturbation-based shift that accounts for the residual-harmonic gradient while leveraging the dominant-harmonic curvature, resulting in consistently smaller errors than the uncorrected estimates.
Increasing the spacing from $\alpha=0.5$ to $\alpha=0.85$ reduces the error for both quantized predictors by increasing the electrical aperture at $f_I$, which sharpens the mainlobe and increases curvature near the maximum, making the peak location less sensitive to multi-harmonic perturbations. 

\begin{figure}
    \centering
    \includegraphics[width=0.95\linewidth]{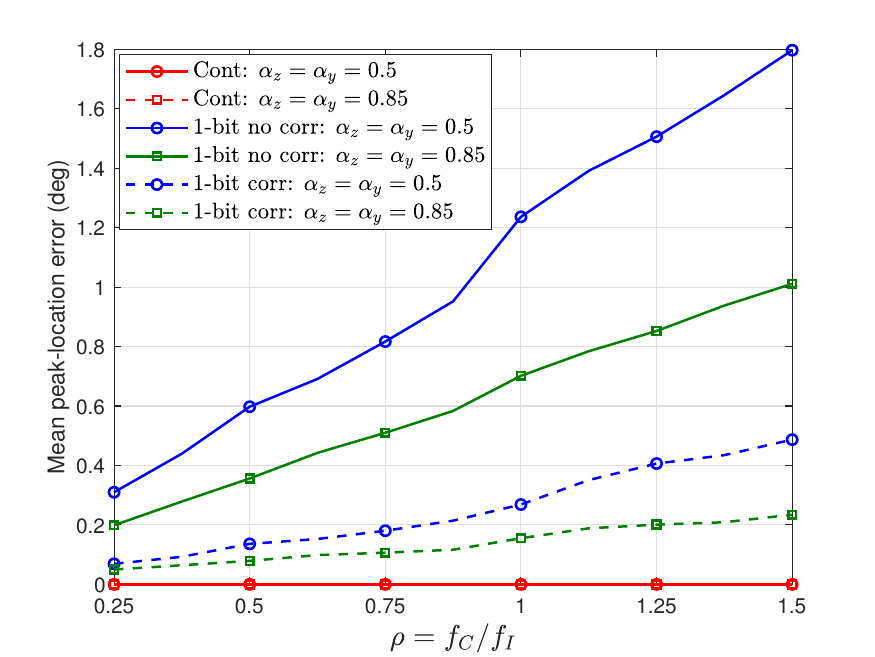}
    \caption{Mean error between the predicted peaks and the actual peaks for both continuous and quantization cases as a function of frequency ratio $\rho$ for different element spacing, and $N_z = N_y = 10$.}
    \label{fig:err-Vs-rho}
\end{figure}

\subsection{Beam-split Existence (Elevation and Azimuth)}

Figs.~\ref{fig:Elevation-Feas}--\ref{fig:Azimuth-Feas} verify the feasibility-based split criteria in {Section~III} for the continuous case.
As shown in Fig.~\ref{fig:Elevation-Feas}, elevation split occurs when the feasible-integer set in \eqref{eq:Elev-Feas} contains at least two integers, i.e., $\lvert\mathcal{M}_z\rvert\ge 2$. {It can be observed that the split probability is essentially zero below $\frac{2\alpha_z}{\rho}=1$, becomes one above $\frac{2\alpha_z}{\rho}=2$, and transitions smoothly in the intermediate region $1\le \frac{2\alpha_z}{\rho}<2$ where the outcome depends on the random interval center induced by $(\varphi_I,\varphi_D)$.}
Fig.~\ref{fig:Azimuth-Feas} shows that, conditioned on a valid elevation peak, azimuth split occurs when the feasible-integer set in \eqref{eq:Azim-Feas} satisfies $\lvert\mathcal{M}_y\rvert\ge 2$. In contrast to elevation, the azimuth feasibility-interval length scales with $A_O=\cos\varphi_O^*$, so the effective length is reduced by the elevation projection and varies across trials. This shifts the high-split region to a larger $\alpha_y$. The overlaid contours indicate where $\frac{2\alpha_y}{\rho}A_O$ exceeds the {boundaries $1$ and $2$ for the feasibility interval in Section~III.C} with probability one (dashed) and in expectation (solid), under the same conditioning.

\begin{figure}
    \centering
    \includegraphics[width=0.95\linewidth]{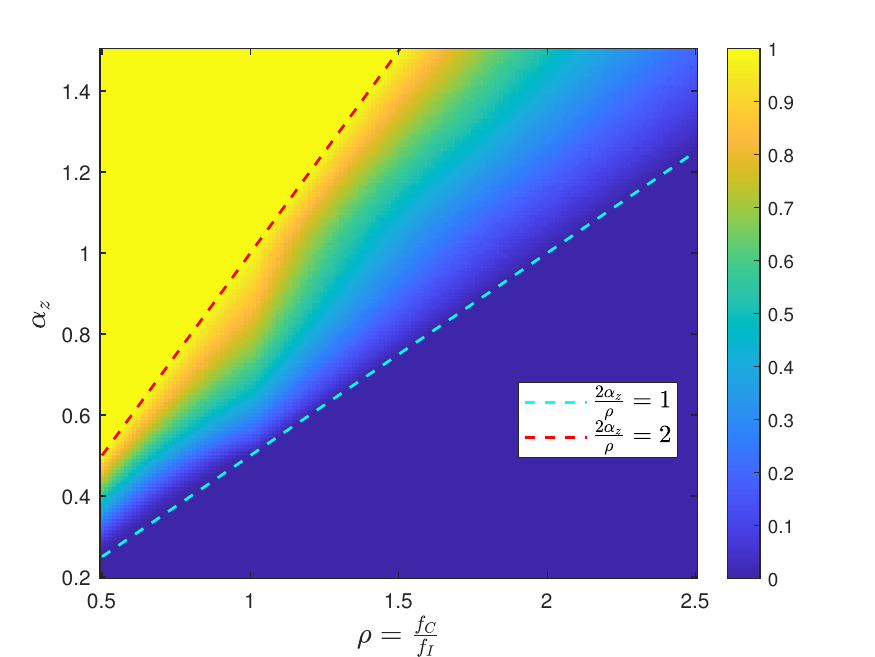}
    \caption{Elevation split probability $\mathbb{P}\!\left(|\mathcal{M}_z|\ge 2\right)$ versus $(\rho,\alpha_z)$, where $\mathcal{M}_z$ is the feasible-integer set in \eqref{eq:Elev-Feas}; dashed lines mark $\frac{2\alpha_z}{\rho}=1$ and $\frac{2\alpha_z}{\rho}=2$.}
    \label{fig:Elevation-Feas}
\end{figure}

\begin{figure}
    \centering
    \includegraphics[width=0.95\linewidth]{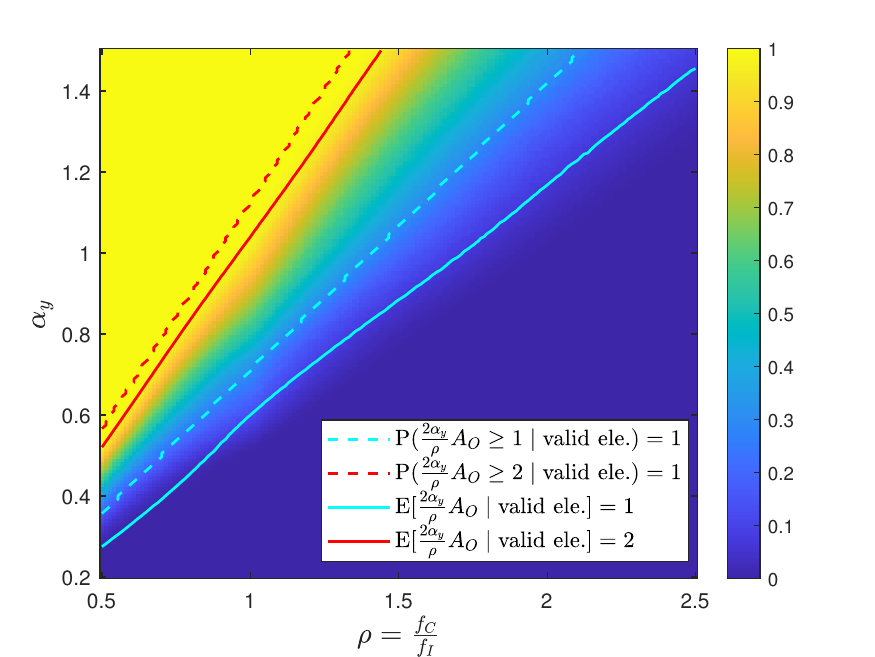}
    \caption{Azimuth split probability $\mathbb{P}\!\left(|\mathcal{M}_y|\ge 2 \mid \text{valid elevation}\right)$ versus $(\rho,\alpha_y)$, where $\mathcal{M}_y$ is the feasible-integer set in \eqref{eq:Azim-Feas}; contours summarize thresholds of $\frac{2\alpha_y}{\rho}A_O$ under the same conditioning.}

    \label{fig:Azimuth-Feas}
\end{figure}




\subsection{Per-frequency Gain and Beam-split Probability}

{
\color{black}

To quantify the performance impact of frequency mismatch, Fig.~\ref{fig:GainLoss_Psplit_Rho} plots the design-direction gain degradation and beam-split probability versus $\rho$. The degradation with respect to the continuous reference is defined as $D_x(\rho)=10\log_{10}\left(\frac{u^{\mathrm{ref}}_{\mathrm{cont}}}{u(\varphi_D,\theta_D,\boldsymbol{\Psi}_{x},\rho)}\right)$, $x\in\{\mathrm{cont},1\text{-bit}\}$, where $u^{\mathrm{ref}}_{\mathrm{cont}}$ is the continuous phase normalized gain at the design direction when $\rho=1$. The beam-split probability is computed as $P_{\mathrm{split}}(\rho)=\Pr\{N_{\mathrm{pk}}(\rho)>1\}$, where $N_{\mathrm{pk}}(\rho)$ is the number of distinct feasible peak angle pairs predicted by the analytical framework.
As shown in Fig.~\ref{fig:GainLoss_Psplit_Rho}, the continuous-phase degradation is zero at $\rho=1$, since the RIS is evaluated at its configuration frequency. As $\rho$ deviates from one, the gain at the original design direction decreases because the maximum-gain lobe moves away from $(\varphi_D,\theta_D)$. The degradation is asymmetric around $\rho=1$: when $\rho<1$, the evaluation frequency is higher than the configuration frequency, so the RIS aperture is effectively larger relative to the wavelength and the reflected lobes become narrower. Hence, a frequency-induced angular displacement causes a sharper gain drop at the design direction. In addition, smaller $\rho$ enlarges the feasible integer intervals in the peak-family conditions, making multiple feasible peaks more likely. Therefore, the continuous RIS exhibits nonzero beam-split probability mainly for $\rho<1$, whereas for $\rho>1$ the degradation is primarily due to beam-squint.

For the $1$-bit RIS, the continuous reference degradation is nonzero at $\rho=1$ because it includes the finite-resolution quantization loss. Moreover, to separate this loss from the additional degradation caused by frequency mismatch, we also plot the 1-bit reference degradation $D_{\mathrm{1bit}}(\rho)=10\log_{10}\left(\frac{u^{\mathrm{ref}}_{\mathrm{1bit}}}{u(\varphi_D,\theta_D,\boldsymbol{\Psi}_{\mathrm{1bit}},\rho)}\right)$, where $u^{\mathrm{ref}}_{\mathrm{1bit}}$ is the 1-bit phase normalized gain at the design direction when $\rho=1$. In particular, the vertical gap between the two 1-bit curves represents the average quantization loss, while the variation of the 1-bit reference curve shows how the quantized profile itself is affected by frequency mismatch. The $1$-bit beam-split probability is much higher than that of the continuous RIS because phase quantization introduces additional harmonic peak families. These harmonic components create extra feasible maximum-gain beams and also modify the gain distribution around the design direction. Therefore, for 1-bit, the continuous reference degradation does not necessarily follow the continuous phase trend. In particular, at some mismatched frequencies, the harmonic components can preserve more gain at the design direction than the shifted continuous phase beams.
Furthermore, since each frequency component in a wideband setting corresponds to a different value of $\rho=f_C/f$, the solid curves can be interpreted as per-frequency design direction degradation, while the dashed curves show the corresponding per-frequency beam-split probability.
}

\begin{figure}
    \centering
    \includegraphics[width=0.95\linewidth]{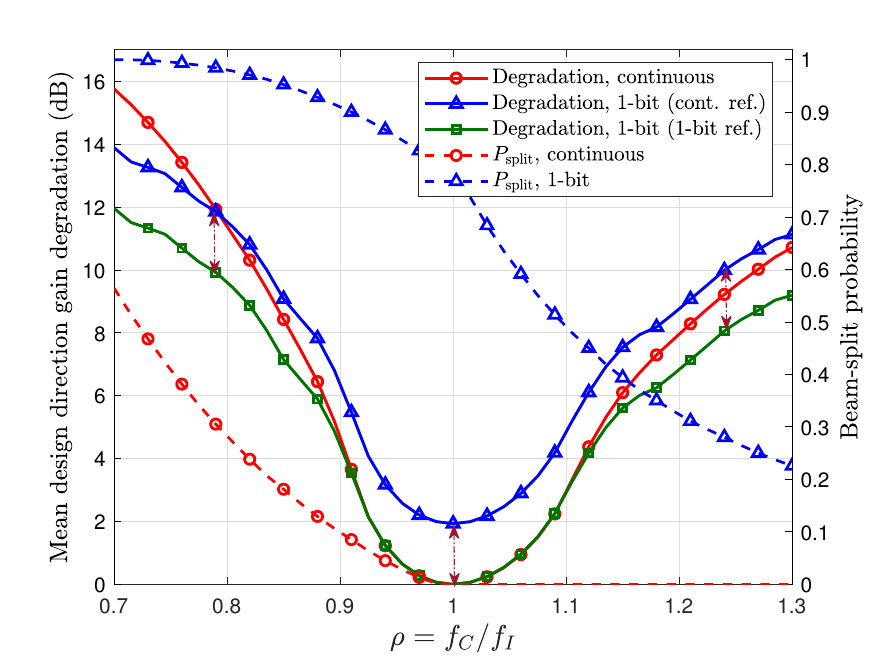}
    \caption{\textcolor{black}{Mean design direction gain degradation and beam-split probability versus frequency ratio $\rho=f_C/f_I$ for $\alpha_y=\alpha_z=0.5$ and $N_y=N_z=12$.}}
    \label{fig:GainLoss_Psplit_Rho}
\end{figure}

\subsection{Number of RIS Elements in $b$-bit Quantization}

\begin{figure}[t]
\centering
\includegraphics[width=0.5\textwidth]{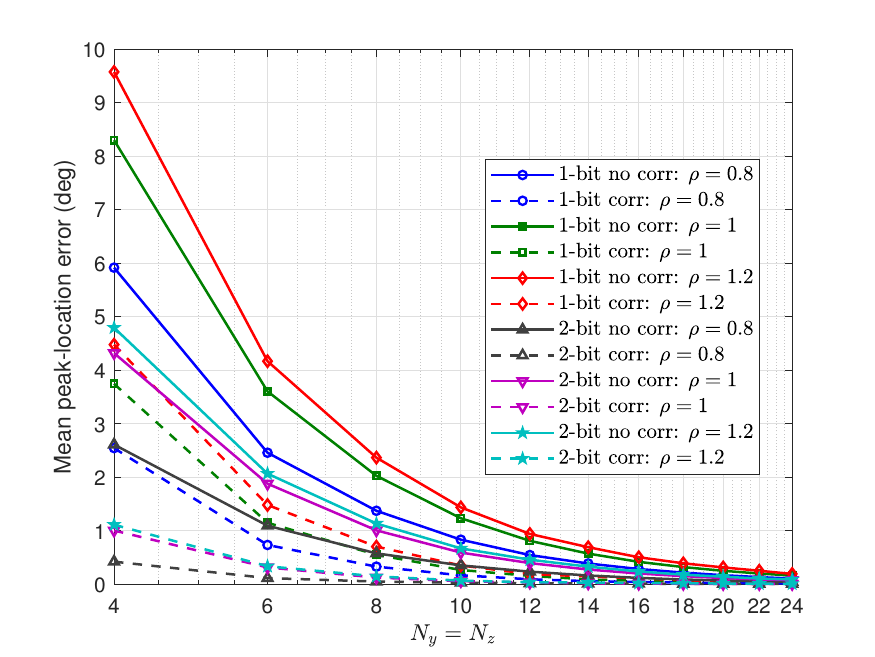}
\caption{\textcolor{black}{Mean error between the predicted peaks and the actual peaks for per-harmonic and corrected peaks versus the number of RIS elements, for different frequency ratio $\rho$.}}
\label{fig:err-Vs-NumElemts}
\end{figure}

\begin{figure}[t]
\centering
\includegraphics[width=0.5\textwidth]{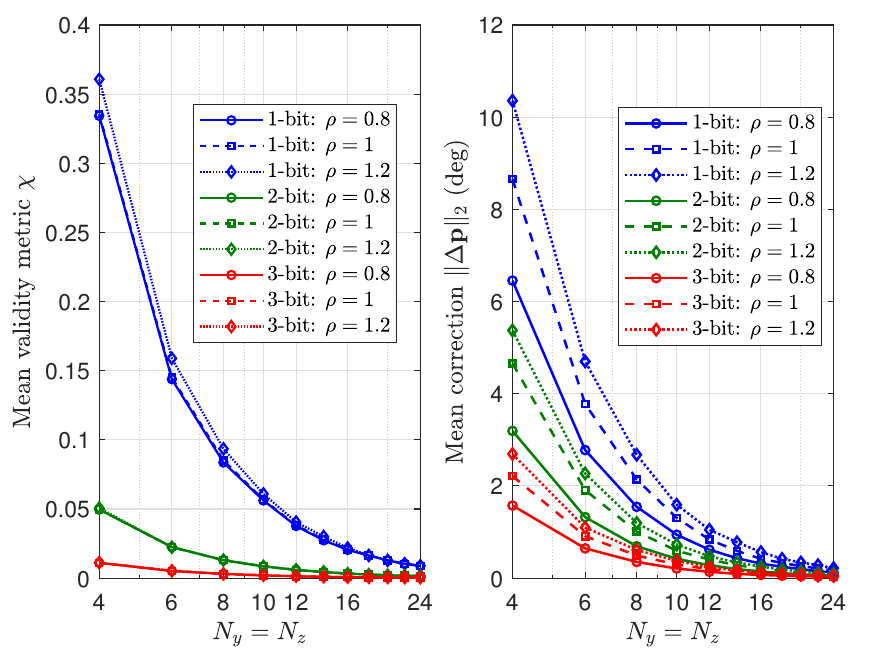}
\caption{\textcolor{black}{Mean validity metric $\chi(\boldsymbol p)$ and angular correction norm $\|\Delta\boldsymbol p\|_2$ versus RIS size $N_y=N_z$ for different quantization resolutions and frequency ratios.}}
\label{fig:chi-step-Vs-NumElemts}
\end{figure}

{\color{black}
Fig.~\ref{fig:err-Vs-NumElemts} plots the mean peak-location error versus
$N_y=N_z$ for different quantization resolutions and frequency ratios, with
and without the proposed correction. For all considered cases, the error
decreases as $N_y=N_z$ increases. This follows from the curvature argument
in Section~V: a larger RIS aperture produces narrower isolated harmonic lobes
and stronger local curvature around each per-harmonic peak. Hence, the
maximizer of the superposed quantized beampattern becomes less sensitive to
the impact of residual harmonics and moves closer to the closed-form per-harmonic
prediction.

The correction gain can be observed from the gap between the ``no corr'' and
``corr'' curves in Fig.~\ref{fig:err-Vs-NumElemts}. The improvement is most
visible for small and moderate RIS sizes, where the harmonic lobes are wider
and the aggregate peak can be displaced more noticeably by the residual
harmonics. As $N_y=N_z$ increases, both the uncorrected and corrected errors
decrease, and the incremental benefit of the correction naturally becomes
smaller. The effect of quantization resolution is also consistent with the
harmonic decomposition. For $1$-bit quantization, the dominant harmonic set
contains two co-dominant harmonics, which increases the impact of the superposition of harmonics. Therefore,
the $1$-bit curves exhibit larger uncorrected errors and a more visible
correction gain. For higher phase resolutions, the non-dominant harmonics are
weaker, so the uncorrected per-harmonic prediction is already more accurate
and the correction mainly acts as a smaller refinement. Finally, the impact of $\rho$ follows the curvature scaling discussed in Section~V:
larger $\rho$ weakens the angle-domain curvature through the $1/\rho^2$
factor and can therefore increase the residual displacement.

Furthermore, Fig.~\ref{fig:chi-step-Vs-NumElemts} quantifies the local validity and angular size of the perturbation correction. The left panel reports the mean validity metric $\chi(\boldsymbol p)$ defined in \eqref{eq:chi-validity-1}. Its decrease with $N_y=N_z$ confirms that the residual-induced perturbation becomes smaller relative to the isolated harmonic lobe as the RIS aperture grows. The metric
is largest for coarse quantization, especially the $1$-bit case, due to
stronger harmonic superposition, and decreases for higher phase resolutions.
The right panel of Fig.~\ref{fig:chi-step-Vs-NumElemts} reports
$\|\Delta\boldsymbol p\|_2$, which is not used as the validity metric but only
as an angular-domain measure of the correction size. This distinction is
important because the same angular shift can have different significance
depending on the local lobe width. The metric $\chi(\boldsymbol p)$ accounts
for the local curvature of the isolated harmonic lobe and therefore measures
locality in the lobe-power sense, whereas $\|\Delta\boldsymbol p\|_2$ shows
the physical angular displacement. The angular correction norm decreases with
$N_y=N_z$ and reflects the effect of $\rho$ more directly: increasing $\rho$
weakens the angle-domain curvature through the $1/\rho^2$ scaling discussed
in Section~V, which can produce a larger correction step.
}

{\color{black}
\subsection{Beam-split-aware Codeword Selection}

Fig.~\ref{fig:AveRate_Vs_NumElemts} shows the average two-user sum-rate versus $N_y=N_z$ under the LoS-dominant multipath channel described in Section~VI. The proposed candidates are generated from the LoS geometry, while the final rate is evaluated using the complete multipath cascaded channel.
We compare four codeword-selection methods within the considered passive finite-resolution RIS architecture, i.e., \textbf{(1)} The home-link-only baseline represents a single-frequency home-band codeword selection, where each RIS selects only the codeword maximizing its intended home cascaded channel at the configuration frequency without accounting for cross-band squint/split.
\textbf{(2)} The squint-aware split-agnostic baseline additionally includes one central cross-band squinted candidate with $(\ell,m_z,m_y)=(0,0,0)$, but ignores split and harmonic lobes. 
\textbf{(3)} The proposed split-aware method includes the home-link codeword and $L_{\mathrm keep}=5$ cross-band candidates generated from the dominant split/harmonic peak families as described in Section~ VI, with $\ell\in\{-1,0\}$ for the considered $1$-bit RIS.
\textbf{(4)} The exhaustive finite-codebook search maximizes the sum-rate over the union of the home- and cross-incidence codebooks at both RISs, and is used as a finite-codebook upper bound. \footnote{\textcolor{black}{TTD-based wideband beamforming architectures are not included because they modify the RIS hardware through frequency-dependent delay elements and therefore fall outside the scope of the passive finite-resolution RIS architecture considered in this work. Accordingly, exhaustive finite-codebook search serves as the upper bound.}
}

It can be observed that the home-link-only baseline has the lowest rate because the cross-band assistance terms in \eqref{eq:Sum-Rate} are not intentionally aligned and arise only incidentally. The squint-aware split-agnostic baseline improves over home-link-only baseline because it includes one central cross-band squinted candidate. However, it cannot exploit the additional feasible peaks created by beam split and phase quantization. The proposed beam-split-aware method further improves the sum-rate because it includes a small set of analytically predicted split/harmonic candidates. These candidates increase the probability that the mismatched-band RIS response is aligned with the corresponding cross user, thereby improving the cross-band cascaded terms in \eqref{eq:Sum-Rate}. The exhaustive finite-codebook search provides the highest rate because it directly searches all codeword pairs in the finite union codebooks and maximizes the actual multipath sum-rate.
The gain often becomes more pronounced for larger RISs because increasing $N_y=N_z$ narrows the reflected beams, making the received power more sensitive to angular misalignment. 

Increasing $Q_{\min}$ from $8$ to $16$ improves all methods by reducing codeword snapping error. At the same time, the relative gaps may decrease because the home-link and central-squint baselines also benefit from the denser codebook. Nevertheless, the proposed method remains above the home-link-only and squint-aware baselines for both resolutions, showing that the gain comes from split/harmonic candidate generation rather than only from codebook resolution. The exhaustive finite-codebook search remains above the proposed method, as expected, because it searches a much larger candidate space. However, the proposed method achieves a significant portion of the exhaustive-search gain with substantially lower candidate-pair search complexity. With $L_{\mathrm{keep}}=5$, each RIS keeps at most one home-link codeword plus five split-aware cross-band codeword, so the proposed method evaluates at most $(1+L_{\mathrm{keep}})^2=36$ phase-profile pairs across the two RISs. In contrast, the exhaustive finite-codebook search considers the union of the home- and cross-incidence codebooks at each RIS. Thus, in the considered two-RIS case, the exhaustive search evaluates up to $(2\times64)^2=16384$ and $(2\times256)^2=262144$ phase-profile pairs for $Q_{\min}=8$ and $Q_{\min}=16$, respectively.

}

\begin{figure}[t]
\centering
\includegraphics[width=0.5\textwidth]{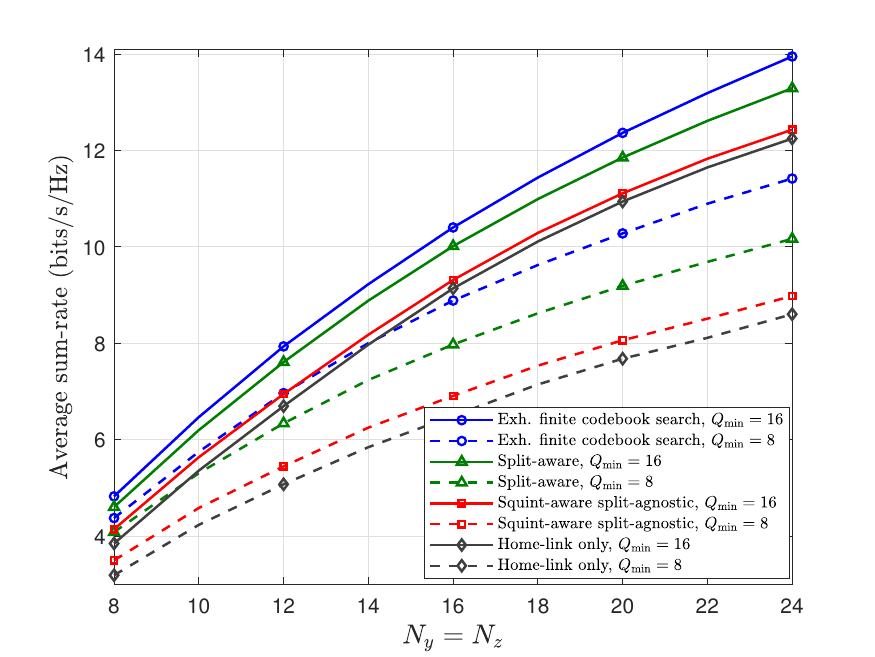}
\caption{\textcolor{black}{Average two-user sum-rate versus RIS size in the two-RIS two-UE multi-band network under LoS-dominant multipath channels. The proposed beam-split-aware selection is compared with home-link only, squint-aware split-agnostic, and exhaustive finite-codebook baselines for $Q_{\min}\in\{8,16\}$, $f_1=30$~GHz, $f_2=40$~GHz, $\rho_{12}=0.75$, and $b=1$.}}
\label{fig:AveRate_Vs_NumElemts}
\end{figure}

\section{Conclusion}
This paper characterized beam-squint and beam-split in UPA-based RISs operating under the configuration--incident frequency mismatch. For continuous phase, we derived necessary and sufficient peak conditions and explicit elevation--azimuth peak families, together with feasibility criteria governed by frequency ratio, element spacing, and incident--design angle pairs. For practical $b$-bit quantization, we reformulated the beampattern via a Fourier-series decomposition of the quantization error phasor, yielding closed-form per-harmonic peak families and existence conditions. To bridge isolated per-harmonic characterizations to the maximizer of the superposed pattern, we developed a perturbation-based angle correction that exploits dominant-harmonic curvature and residual-harmonic gradients. Numerical results validated the peak/split characterizations, quantified the impact of frequency ratio, element spacing, RIS size, and quantization depth on angular displacement, and show that beam-split-aware codeword selection improves sum-rate in a two-RIS two-UE MBN.

\section*{Appendix A: Proof of \textbf{Lemma~\ref{lem:PhaseDiff}}}
\label{Appendix-A}
    The necessity part of Lemma~\ref{lem:PhaseDiff} follows directly from $\phi_{n_y,n_z} \equiv \phi_{0,0} \pmod{2\pi}, \ \forall n_y,n_z$, since the condition holds for all elements. 
For the sufficiency, assume there exists fixed integers $m_y$ and $m_z$ such that $\phi_{n_y+1,n_z}-\phi_{n_y,n_z}=2\pi m_y$ and $\phi_{n_y,n_z+1}-\phi_{n_y,n_z}=2\pi m_z$ hold for all $(n_y,n_z)$. 
    First, fix $n_z = 0$. The relation for unit-step increments along the $y$-direction becomes: $\phi_{n_y+1,0}=\phi_{n_y,0} + 2\pi m_y$. Applying this relation repeatedly from $n_y = 0$, after $n_y$ steps we obtain: $
        \phi_{n_y,0} = \phi_{0,0} + n_y(2\pi m_y).$
    Now, fix $n_y$ and start from $n_z = 0$. Applying the unit-step increment along the $z$-direction yields:
    \begin{equation}\label{eq:Lemma1-z-increm}
        \phi_{n_y,n_z} = \phi_{n_y,0} + n_z(2\pi m_z).
    \end{equation}
    Substituting $\phi_{n_y,0}$ into \eqref{eq:Lemma1-z-increm}, we obtain:
    \begin{equation}
        \phi_{n_y,n_z} = \phi_{0,0} + 2\pi(n_y m_y + n_z m_z).
    \end{equation}
    Since $n_y m_y + n_z m_z$ is always an integer, it follows 
 $\phi_{n_y,n_z} \equiv \phi_{0,0} \pmod{2\pi}$, thus proving the sufficiency part of Lemma~\ref{lem:PhaseDiff}.
 

\section*{Appendix B: Proof of \textbf{Lemma~\ref{lem:FourierConve}}}
\label{Appendix-B}

    \textit{Periodicity of $g(x)$}: For the round function $\lfloor x\rceil$ and any integer $k$, we have $\lfloor x\rceil = k$ for $x \in [k - \frac{1}{2}, k + \frac{1}{2})$, with a fixed (but arbitrary) convention at the endpoints $x = k \pm \frac{1}{2}$. This implies that $\lfloor x + 1\rceil = \lfloor x\rceil + 1$ for all $x$. Hence, 
    \begin{equation}
        \varepsilon(x + 1) = \Delta(\lfloor x + 1\rceil - (x + 1)) = \Delta(\lfloor x\rceil - x) = \varepsilon(x),
    \end{equation}
    which shows that $\varepsilon(x)$ is 1-periodic, and so is $g(x) = e^{j\varepsilon(x)}$.

    \textit{Bounded variation of $g(x)$ on one period}: Let us define total variation and bounded variation:
    \begin{definition}
        The variation $V_{[a,b]}(f)$ of a function $f: [a,b] \rightarrow \mathbb{C}$ is 
        $V_{[a,b]}(f) = \sup_{a = x_0 < \dots < x_m = b} \sum_{i=1}^m \abs{f(x_i) - f(x_{i-1})}$.
        A function $f$ is said to be of bounded variation on $[a,b]$ if $V_{[a,b]}(f) < \infty$ \cite{EvansGariepy2015}.
    \end{definition}

    Now, we show that $g(x)$ is of bounded variation on the period $[0,1]$. We use the following identity, which follows from trigonometric identities and the inequality $\abs{\sin x} \leq \abs{x}$:
    \begin{equation}\label{eq:trig-ident}
        \abs{e^{ju} - e^{jv}} = 2\abs{\sin\left(\frac{u - v}{2}\right)} \leq \abs{u - v}.
    \end{equation}
    Partition the interval $[0,1]$ at $x = \frac{1}{2}$. For $x \in [0, \frac{1}{2})$, we have $g(x) = e^{-j\Delta x}$, so
    \begin{align}
        \sum\limits_{x_i \in [0, \frac{1}{2})} \abs{g(x_i) - g(x_{i-1})}
        & \overset{(a)}{\leq} \sum\limits_{x_i \in [0, \frac{1}{2})} \abs{-\Delta (x_i - x_{i-1})} \notag \\
        & \overset{(b)}{=} \abs{-\Delta \left(\tfrac{1}{2}^- - 0\right)} = \tfrac{1}{2} \abs{\Delta},
    \end{align}
    where (a) follows from \eqref{eq:trig-ident} and (b) holds because $-\Delta x$ is monotonic.
    For $x \in (\frac{1}{2}, 1]$, we have $g(x) = e^{j\Delta(1 - x)}$, and similarly, $
    \sum\limits_{x_i \in (\frac{1}{2},1]} \abs{g(x_i) - g(x_{i-1})} \leq \tfrac{1}{2} \abs{\Delta}.$
    At the jump point $x = \frac{1}{2}$, we have $g(\frac{1}{2}^-) = e^{-j\frac{\Delta}{2}}$ and $g(\frac{1}{2}^+) = e^{j\frac{\Delta}{2}}$, so $
        \abs{g(\tfrac{1}{2}^+) - g(\tfrac{1}{2}^-)} = 2 \abs{\sin \tfrac{\Delta}{2}}.$
    Summing across all partitions, we obtain: $\sum_{i=1}^m \abs{g(x_i) - g(x_{i-1})} \leq \abs{\Delta} + 2\abs{\sin \tfrac{\Delta}{2}} \leq 2\abs{\Delta} < \infty.$
    Therefore, $g(x)$ is of bounded variation over the period $[0,1]$.

    \textit{Convergence of the Fourier series}: To show that the Fourier series of $g(x)$ converges, we use the Dirichlet–Jordan test:
    \begin{theorem}[Dirichlet–Jordan Test]\label{theo:Diri-Jord}
        If a $2\pi$-periodic function $f:\mathbb{R} \rightarrow \mathbb{C}$ is of bounded variation  over $[0,2\pi]$, then its Fourier series converges at every point $x$ to $\frac{1}{2}(f(x^+) + f(x^-))$ \cite{berg2009fourier}.
    \end{theorem}

    Define $2\pi$-periodic function $G(y) = g\left(\frac{y}{2\pi}\right)$. Then, by applying the Dirichlet–Jordan test to $G(y)$ and re-scaling to $g(x)$, we conclude that the Fourier series of $g(x)$ converges pointwise to $\tilde{g}(x) = \frac{1}{2}(g(x^+) + g(x^-))$.
    At continuity points, $\tilde{g}(x) = g(x)$, and at the jump points $x = \frac{1}{2} + k$, the sum equals $\tilde{g}(x) = \cos \frac{\Delta}{2}$. This completes the proof of \textbf{Lemma}~\ref{lem:FourierConve}.
    

\section*{Appendix C: Proof of \textbf{Lemma~\ref{lem:FourierSeriesCoeff}}}
\label{Appendix-C}

On $x \in [0,1]$, we can reformulate $g(x)$ as:
\begin{equation}
g(x)=
\begin{cases}
e^{-j\Delta x}, & 0\le x<\tfrac12,\\
e^{j\Delta}\,e^{-j\Delta x}, & \tfrac12\le x\leq1.
\end{cases}
\end{equation}
Fix $\ell\in\mathbb{Z}$ and set $a=\Delta+2\pi \ell \neq 0$. Then the Fourier series coefficient of $g(x)$ is obtained by:
\begin{align}
\widehat g[\ell]
&=\int_{0}^{\tfrac12} e^{-j a x} dx+e^{j\Delta} \int_{\tfrac12}^{1} e^{-j a x} dx
=\frac{1-e^{-j\tfrac{a}{2}}}{ja} \notag \\ & +e^{j\Delta}\frac{e^{-j\tfrac{a}{2}}-e^{-ja}}{ja}.
\end{align}
Since $e^{-ja}=e^{-j(\Delta+2\pi\ell)}=e^{-j\Delta}$,
\begin{align}
\widehat g[\ell]
\!=\!\frac{e^{-j\tfrac{a}{2}}\!\big(e^{j\Delta}\!-\!1\big)}{ja}
\!=\!\frac{2j e^{-j\tfrac{\Delta+2\pi\ell-\Delta}{2}}\!\!\sin\tfrac{\Delta}{2}}{ja} \!=\!\frac{2(-1)^{\ell}\sin\tfrac{\Delta}{2}}{\Delta+2\pi \ell}. 
\end{align}
Then, using $\widehat g[\ell]$, $g(x)$ can be represented by its Fourier series $g(x)\sim \sum_{\ell\in\mathbb{Z}}\widehat g[\ell]\;e^{j2\pi \ell x}$.
Moreover, followed by \textbf{Lemma~\ref{lem:FourierConve}}, this Fourier series converges pointwise to $\tfrac12\big(g(x^-)+g(x^+)\big)$.

{\color{black}
\section*{Appendix D: Negative Definiteness of $\boldsymbol H_{J_\ell}(\boldsymbol p)$}
\label{Appendix-D}

For the $\ell$-th harmonic, $J_\ell = \abssOne{S_\ell}^2$ can be written as:
\begin{equation}
J_\ell(\varphi_O,\theta_O)
=
\abssSq{\widehat g[\ell]}
F_{N_z}(\beta_z(\ell))
F_{N_y}(\beta_y(\ell)),
\end{equation}
 where
$F_N(\beta)=\abssSq{\sum_{n=0}^{N-1} e^{j2\pi\beta n}}.$ Moreover, at any integer $m$, we obtain $F^\prime(m) = 0$ and $F^{\prime\prime} = -\frac{2\pi^2}{3}N^2(N^2-1),$ such that $F^{\prime\prime} <0 $ for $N \geq 2$.

On the other hand, at the feasible per-harmonic peak $\boldsymbol p$, the peak conditions in \eqref{eq:z-peak-cond-Quantization} and \eqref{eq:y-peak-cond-Quantization} give
$\beta_z(\ell)\big|_{\boldsymbol p}=m_z$ and
$\beta_y(\ell)\big|_{\boldsymbol p}=m_y.$
Now consider $J_\ell$ first as a function of the two phase-slope variables
$\boldsymbol\beta_\ell=[\beta_z(\ell),\beta_y(\ell)]^T$. Since $F'_{N_z}(m_z)=0$ and $F'_{N_y}(m_y)=0$, the Hessian of $J_\ell$ with respect to $\boldsymbol\beta_\ell$ at $(m_z,m_y)$ is diagonal and given by
\begin{equation}
\boldsymbol H_{\beta}
=
-\frac{2\pi^2}{3}
\abssSq{\widehat g[\ell]}N_z^2N_y^2
\begin{bmatrix}
N_z^2-1 & 0\\
0 & N_y^2-1
\end{bmatrix}.
\end{equation}
Therefore,
$\boldsymbol H_{\beta}\prec0,$
for $N_z,N_y\geq2$ and $\widehat g[\ell]\neq0$.

From the definitions of $\beta_z(\ell)$ and $\beta_y(\ell)$, the Jacobian of $\boldsymbol\beta_\ell$ with respect to $(\varphi_O,\theta_O)$ at $\boldsymbol p=[\varphi_O^*,\theta_O^*]^T$ is:
\begin{equation}
\boldsymbol G
= \frac{\partial \boldsymbol \beta_\ell}{\partial\boldsymbol p} =
\begin{bmatrix}
\frac{\alpha_z}{\rho}\cos\varphi_O^* & 0\\
-\frac{\alpha_y}{\rho}\sin\theta_O^*\sin\varphi_O^* &
\frac{\alpha_y}{\rho}\cos\theta_O^*\cos\varphi_O^*
\end{bmatrix},
\end{equation}
with determinant 
$\det(\boldsymbol G)
=
\frac{\alpha_z\alpha_y}{\rho^2}
\cos\theta_O^*\cos^2\varphi_O^*.$
For $\alpha_z,\alpha_y>0$, $\rho>0$, and observation angles
$\varphi_O^*,\theta_O^*\in(-\frac{\pi}{2},\frac{\pi}{2})$, we have
$\det(\boldsymbol G)>0$. Hence, $\boldsymbol G$ has full rank.

By the second-order chain rule \cite{petersen_pedersen_2012},
\begin{equation}
\boldsymbol H_{J_\ell}(\boldsymbol p)
=
\boldsymbol G^T\boldsymbol H_{\beta}\boldsymbol G
+
\sum_{i\in\{z,y\}}
\frac{\partial J_\ell}{\partial \beta_i}
\boldsymbol H_{\beta_i}(\boldsymbol p),
\end{equation}
where $\boldsymbol H_{\beta_i}(\boldsymbol p)$ denotes the Hessian of $\beta_i(\ell)$ with respect to $(\varphi_O,\theta_O)$. However, at the per-harmonic peak, $\frac{\partial J_\ell}{\partial \beta_z}=0$ and $\frac{\partial J_\ell}{\partial \beta_y}=0$, so $
\boldsymbol H_{J_\ell}(\boldsymbol p)
=
\boldsymbol G^T\boldsymbol H_{\beta}\boldsymbol G.$
For any nonzero vector $\boldsymbol x\in\mathbb R^2$, the full rank of $\boldsymbol G$ implies $\boldsymbol G\boldsymbol x\neq\boldsymbol0$. Since $\boldsymbol H_{\beta}\prec0$, we obtain $
\boldsymbol x^T\boldsymbol H_{J_\ell}(\boldsymbol p)\boldsymbol x
=
(\boldsymbol G\boldsymbol x)^T
\boldsymbol H_{\beta}
(\boldsymbol G\boldsymbol x)
<0.$
Hence, $
\boldsymbol H_{J_\ell}(\boldsymbol p)\prec0.$
This proves that the isolated per-harmonic peak has strictly negative curvature in the angular domain under the stated conditions.

}
\section*{\textcolor{black}{Appendix E}: Gradient and Hessian of $J$ used in \eqref{eq:shift-Estimate}}
\label{Appendix-E}

Let $f(\varphi,\theta)\in\mathbb{C}$ and define $J_f(\varphi,\theta)=|f(\varphi,\theta)|^2$.
Then
\begin{align}
J_{f,\varphi} = 2\Re\{ f_{\varphi} f^*\},\quad
J_{f,\theta}  = 2\Re\{ f_{\theta} f^*\}.
\end{align}
The Hessian entries are: $J_{f,\varphi\varphi} = 2\Re\{ f_{\varphi\varphi} f^* + f_{\varphi} f_{\varphi}^*\}$, $J_{f,\theta\theta}  = 2\Re\{ f_{\theta\theta} f^* + f_{\theta} f_{\theta}^*\}$, and $J_{f,\varphi\theta} = 2\Re\{ f_{\varphi\theta} f^* + f_{\varphi} f_{\theta}^*\}$
with $J_{f,\theta\varphi}=J_{f,\varphi\theta}$.
Hence,
$\boldsymbol H_{J_f}(\varphi,\theta)=
\begin{bmatrix}
J_{f,\varphi\varphi} & J_{f,\varphi\theta}\\
J_{f,\varphi\theta} & J_{f,\theta\theta}
\end{bmatrix}.$

With $J=|S|^2$ and $J_\ell=|S_\ell|^2$, we have:
\begin{align}
\nabla J &=
2\begin{bmatrix}
\Re\{S_{\varphi}S^*\}\\
\Re\{S_{\theta}S^*\}
\end{bmatrix},\qquad
\nabla J_\ell =
2\begin{bmatrix}
\Re\{S_{\ell,\varphi}S_\ell^*\}\\
\Re\{S_{\ell,\theta}S_\ell^*\}
\end{bmatrix},
\end{align}
and
\begin{align}
\boldsymbol H_{J} &=
2\begin{bmatrix}
\Re\{S_{\varphi\varphi}S^*+S_{\varphi}S_{\varphi}^*\} &
\Re\{S_{\varphi\theta}S^*+S_{\varphi}S_{\theta}^*\}\\
\Re\{S_{\varphi\theta}S^*+S_{\varphi}S_{\theta}^*\} &
\Re\{S_{\theta\theta}S^*+S_{\theta}S_{\theta}^*\}
\end{bmatrix},\\
\boldsymbol H_{J_\ell} &=
2\begin{bmatrix}
\Re\{S_{\ell,\varphi\varphi}S_\ell^*+S_{\ell,\varphi}S_{\ell,\varphi}^*\} &
\Re\{S_{\ell,\varphi\theta}S_\ell^*+S_{\ell,\varphi}S_{\ell,\theta}^*\}\\
\Re\{S_{\ell,\varphi\theta}S_\ell^*+S_{\ell,\varphi}S_{\ell,\theta}^*\} &
\Re\{S_{\ell,\theta\theta}S_\ell^*+S_{\ell,\theta}S_{\ell,\theta}^*\}
\end{bmatrix}.
\end{align}
Using $\nabla R_\ell(\boldsymbol p)=\nabla J(\boldsymbol p)$, we obtain $\nabla R_\ell(\boldsymbol p)$ using
$\nabla R_\ell=\nabla J-\nabla J_\ell.$
In particular, at the per-harmonic nominal peak $\boldsymbol p$ of $J_\ell$,
$\nabla J_\ell(\boldsymbol p)=\boldsymbol 0$, and therefore
$\nabla R_\ell(\boldsymbol p)=\nabla J(\boldsymbol p)$.

In \eqref{eq:beam-patt-err}, define $\Phi_\ell
=2\pi\Big(
\beta_z^{(\ell)}(\boldsymbol{p})\,n_z
+\beta_y^{(\ell)}(\boldsymbol{p})\,n_y
\Big)$. The first derivatives of the phase slopes are:
$\beta_{z,\varphi} = \tfrac{\alpha_z}{\rho}\cos\varphi_O$, $\beta_{z,\theta} = 0$, $\beta_{y,\varphi} = -\tfrac{\alpha_y}{\rho}\sin\theta_O\sin\varphi_O$, and 
$\beta_{y,\theta} = \tfrac{\alpha_y}{\rho}\cos\theta_O\cos\varphi_O$.
The second derivatives are:
$\beta_{z,\varphi\varphi} = -\tfrac{\alpha_z}{\rho}\sin\varphi_O$, 
$\beta_{z,\theta\theta} = 0$, $\beta_{z,\varphi\theta} = 0$, $\beta_{y,\varphi\varphi} = -\tfrac{\alpha_y}{\rho}\sin\theta_O\cos\varphi_O$, $\beta_{y,\theta\theta} = -\tfrac{\alpha_y}{\rho}\sin\theta_O\cos\varphi_O$, and $\beta_{y,\varphi\theta} = -\tfrac{\alpha_y}{\rho}\cos\theta_O\sin\varphi_O$.
Thus, $ \Phi_{\ell,\varphi} = 2\pi\big(n_z\beta_{z,\varphi}+n_y\beta_{y,\varphi}\big)$, 
$\Phi_{\ell,\theta} = 2\pi\big(n_z\beta_{z,\theta}+n_y\beta_{y,\theta}\big)$,
$
\Phi_{\ell,\varphi\varphi} = 2\pi\big(n_z\beta_{z,\varphi\varphi}+n_y\beta_{y,\varphi\varphi}\big)$,
$
\Phi_{\ell,\theta\theta} = 2\pi\big(n_z\beta_{z,\theta\theta}+n_y\beta_{y,\theta\theta}\big)$, and
$
\Phi_{\ell,\varphi\theta} = 2\pi\big(n_z\beta_{z,\varphi\theta}+n_y\beta_{y,\varphi\theta}\big)$.

Using $\partial_x e^{j\Phi}=j\Phi_x e^{j\Phi}$ and
$\partial_{xx} e^{j\Phi}=(j\Phi_{xx}-\Phi_x^2)e^{j\Phi}$,
$\partial_{xy} e^{j\Phi}=(j\Phi_{xy}-\Phi_x\Phi_y)e^{j\Phi}$,
we obtain:
\begin{equation}
S_{\ell,\varphi} =
\widehat g[\ell]\!\!\sum_{n_z=0}^{N_z-1}\sum_{n_y=0}^{N_y-1}
\big(j\Phi_{\ell,\varphi}\big)e^{j\Phi_\ell},
\end{equation}
\begin{equation}
S_{\ell,\theta} =
\widehat g[\ell]\!\!\sum_{n_z=0}^{N_z-1}\sum_{n_y=0}^{N_y-1}
\big(j\Phi_{\ell,\theta}\big)e^{j\Phi_\ell},
\end{equation}
\begin{equation}
S_{\ell,\varphi\varphi} =
\widehat g[\ell]\!\!\sum_{n_z=0}^{N_z-1}\sum_{n_y=0}^{N_y-1}
\big(j\Phi_{\ell,\varphi\varphi}-\Phi_{\ell,\varphi}^2\big)e^{j\Phi_\ell},
\end{equation}
\begin{equation}
S_{\ell,\theta\theta} =
\widehat g[\ell]\!\!\sum_{n_z=0}^{N_z-1}\sum_{n_y=0}^{N_y-1}
\big(j\Phi_{\ell,\theta\theta}-\Phi_{\ell,\theta}^2\big)e^{j\Phi_\ell},
\end{equation}
\begin{equation}
S_{\ell,\varphi\theta} =
\widehat g[\ell]\!\!\sum_{n_z=0}^{N_z-1}\sum_{n_y=0}^{N_y-1}
\big(j\Phi_{\ell,\varphi\theta}-\Phi_{\ell,\varphi}\Phi_{\ell,\theta}\big)e^{j\Phi_\ell}.
\end{equation}


\bibliography{ref}
\bibliographystyle{IEEEtran}

\begin{IEEEbiography}[{\includegraphics[width=1in,height=1.25in, clip,keepaspectratio]{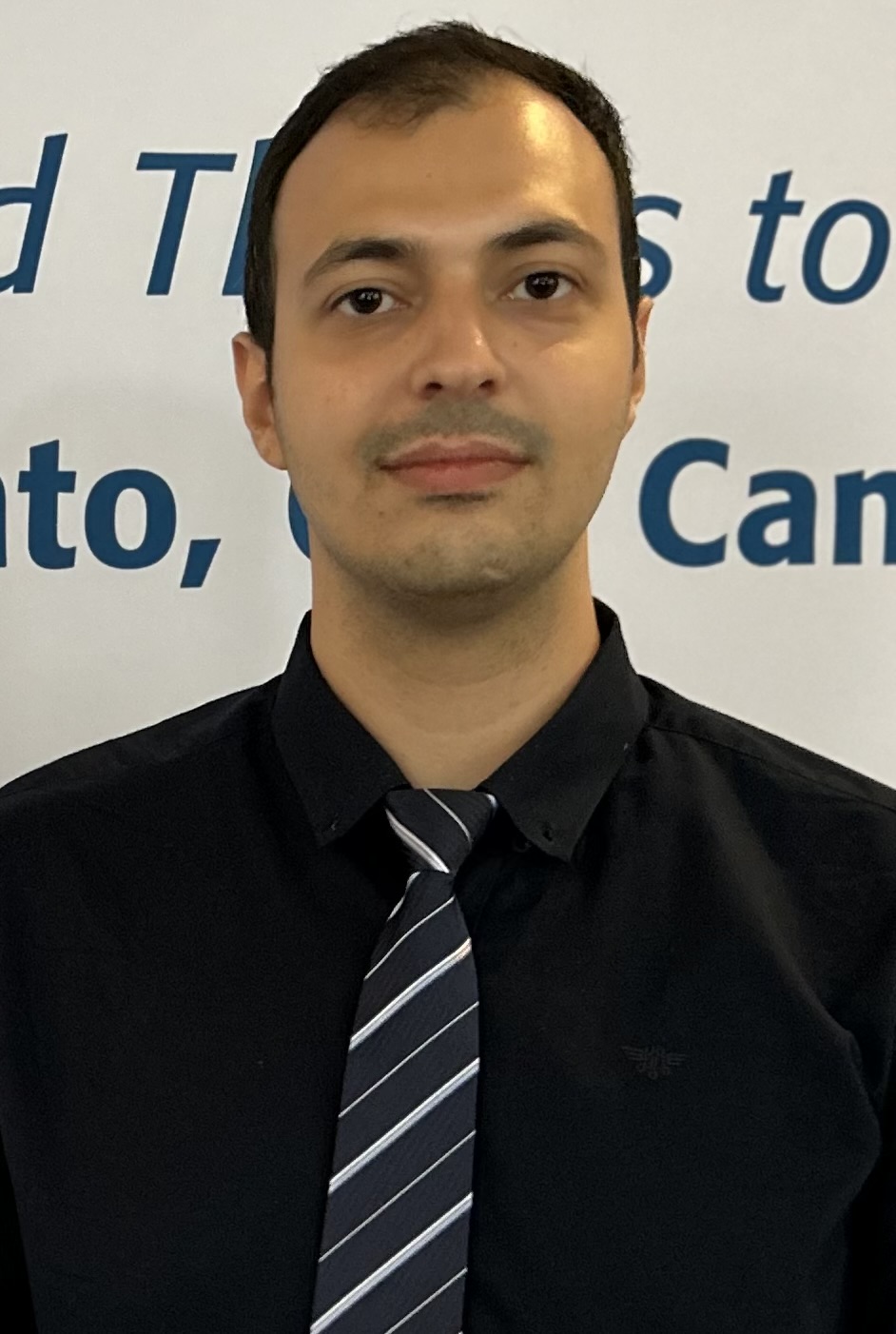}}]{\textbf{Mohammad Amin Saeidi}} (Member, IEEE) received the M.Sc. degree in Electrical Engineering–Communication Systems from Amirkabir University of Technology, Tehran, Iran, in 2021, and the Ph.D. degree in Electrical Engineering and Computer Science from York University, Canada, in 2026. He is currently a Postdoctoral Fellow at the University of Toronto, Canada. His research interests include quantum sensing, 6G wireless communications, reconfigurable intelligent surfaces, near-field communications, integrated sensing and communications, multi-band and terahertz communications, mobility management, and optimization. He has served as a reviewer for several IEEE journals, including the IEEE Transactions on Wireless Communications, IEEE Transactions on Communications, IEEE Transactions on Mobile Computing, IEEE Journal on Selected Areas in Communications, IEEE Open Journal of the Communications Society, IEEE Wireless Communications Letters, IEEE Communications Letters, and IEEE Transactions on Green Communications and Networking.
\end{IEEEbiography}

\begin{IEEEbiography}
[{\includegraphics[width=1in,height=1.2in,clip,keepaspectratio]{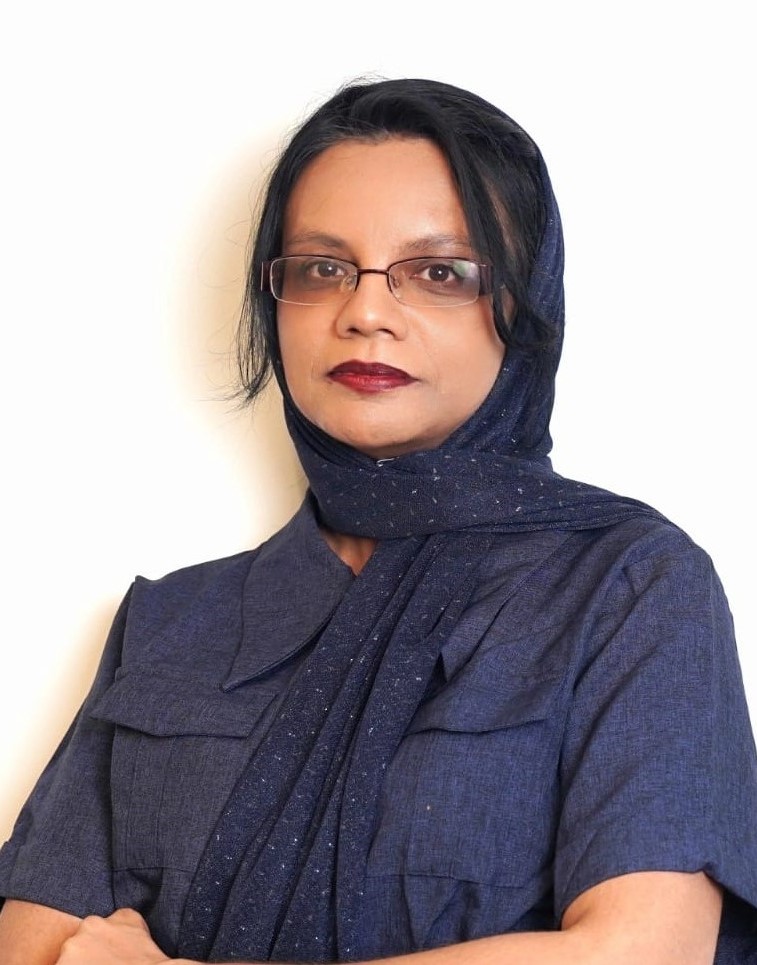}}]
{Hina Tabassum}
 (Senior Member, IEEE) received the Ph.D. degree from the King Abdullah University of Science and Technology (KAUST). She is currently an Associate Professor with the Lassonde School of Engineering, York University, Canada, where she joined in 2018. She holds the York Research Chair in 5G/6G-Enabled Mobility and Sensing Applications, serves as a Visiting Faculty at the University of Toronto, and was appointed Graduate Program Director of Electrical Engineering and Computer Science at York University in 2026. Her research focuses on the modeling, analysis, and optimization of next-generation wireless communication, localization, and sensing networks. She has co-authored two books, eight book chapters, and over 130 publications. She has delivered more than 20 invited tutorials and talks at international venues on wireless communications and sensing networks. She was selected as an IEEE Communications Society Distinguished Lecturer for 2025–2026 and has been listed among Stanford’s World’s Top 2\% Researchers from 2021 to 2025. Her recognitions include the N2Women STAR (2025), the Early Career Lassonde Innovation Award (2023), and the N2Women Rising Star Award (2022). She is currently an Area Editor for IEEE Open Journal of the Communications Society and IEEE Communications Surveys \& Tutorials, and an Associate Editor for IEEE Transactions on Communications and IEEE Transactions on Wireless Communications.
\end{IEEEbiography}

\end{document}